\documentclass[journal]{IEEEtran}
\usepackage{amsmath}
\usepackage{amssymb}
\usepackage{amsfonts}
\usepackage{graphicx}
\usepackage{epsfig}
\usepackage{subfigure}
\usepackage{psfrag}
\usepackage{color}
\usepackage{url}
\usepackage{cite}

\begin{document}
\title{Energy Beamforming with One-Bit Feedback}

\author{Jie Xu and Rui Zhang
\thanks{Part of this paper has been presented in the IEEE International Conference on Acoustics, Speech, and Signal Processing (ICASSP), Florence, Italy, May 4-9, 2014 \cite{XuJie2014}.}
\thanks{J. Xu is with the Department of
Electrical and Computer Engineering, National University of
Singapore (e-mail: elexjie@nus.edu.sg).}
\thanks{R. Zhang is with
the Department of Electrical and Computer Engineering, National
University of Singapore (e-mail: elezhang@nus.edu.sg). He is also
with the Institute for Infocomm Research, A*STAR, Singapore.}}
\maketitle

\begin{abstract}
Wireless energy transfer (WET) has attracted significant attention recently for delivering energy to electrical devices without the need of wires or power cables. In particular, the radio frequency (RF) signal enabled far-field WET is appealing to power energy-constrained wireless networks in a broadcast manner. To overcome the significant path loss over wireless channels, multi-antenna or multiple-input multiple-output (MIMO) techniques have been proposed to enhance both the transmission efficiency and range for RF-based WET. However, in order to reap the large energy beamforming gain in MIMO WET, acquiring the channel state information (CSI) at the energy transmitter (ET) is an essential task. This task is particularly challenging for WET systems, since existing channel training and feedback methods used for communication receivers may not be implementable at the energy receiver (ER) due to its hardware limitation. To tackle this problem, we consider in this paper a multiuser MIMO WET system, and propose a new channel learning method that requires only one feedback bit from each ER to the ET per feedback interval. Specifically, each feedback bit indicates the increase or decrease of the harvested energy by each ER in the present as compared to the previous intervals, which can be measured without changing the existing structure of the ER. Based on such feedback information, the ET adjusts transmit beamforming in subsequent training intervals and at the same time obtains improved estimates of the MIMO channels to different ERs by applying an optimization technique called analytic center cutting plane method (ACCPM). For the proposed ACCPM based channel learning algorithm, we analyze its worst-case convergence, from which it is revealed that the algorithm is able to estimate multiuser MIMO channels at the same time without reducing the analytic convergence speed. Furthermore, through extensive simulations, we show that the proposed algorithm outperforms existing one-bit feedback based channel learning schemes in terms of both convergence speed and energy transfer efficiency, especially when the number of ERs becomes large.
\end{abstract}
\begin{keywords}
Wireless energy transfer (WET), multiple-input multiple-output (MIMO), energy beamforming, channel learning, one-bit feedback, analytic center cutting plane method (ACCPM).
\end{keywords}

\IEEEpeerreviewmaketitle
\newtheorem{definition}{\underline{Definition}}[section]
\newtheorem{fact}{Fact}
\newtheorem{assumption}{Assumption}
\newtheorem{theorem}{\underline{Theorem}}[section]
\newtheorem{lemma}{\underline{Lemma}}[section]
\newtheorem{corollary}{\underline{Corollary}}[section]
\newtheorem{proposition}{\underline{Proposition}}[section]
\newtheorem{example}{\underline{Example}}[section]
\newtheorem{remark}{\underline{Remark}}[section]
\newtheorem{algorithm}{\underline{Algorithm}}[section]
\newcommand{\mv}[1]{\mbox{\boldmath{$ #1 $}}}

\section{Introduction}\label{sec:introduction}
\PARstart{W}{ireless} energy transfer (WET) has attracted significant interests recently for delivering energy to electrical devices over the air. Generally, WET can be implemented by inductive coupling via magnetic field induction, magnetic resonant coupling based on the principle of resonant coupling, or electromagnetic (EM) radiation. The different types of WET techniques in practice have their respective advantages and disadvantages (see e.g. \cite{Xie2013} and the references therein). For example, inductive coupling and magnetic resonant coupling both have high energy transfer efficiency for short-range (e.g., several centimeters) and mid-range (say, a couple of meters) applications, respectively; however, it is difficult to apply them to charge freely located devices simultaneously. In contrast, EM radiation based far-field WET, particularly over the radio frequency (RF) bands, is applicable for much longer range (up to tens of meters) applications and also capable of charging multiple devices even when they are moving by exploiting the broadcast nature of RF signal propagation; whereas its energy transfer efficiency may fall rapidly over distance.

RF signal enabled WET is anticipated to have abundant applications in providing cost-effective and perpetual energy supplies to energy-constrained wireless networks such as sensor networks in future. In fact, applying RF-based WET in various types of wireless communication networks has been extensively studied in the literature recently. In general, there are two main lines of research that have been pursued, namely simultaneous wireless information and power transfer (SWIPT) (see e.g. \cite{ZhangHo2013,Liu2013,ZhouZhangHo2013,XuLiuZhang2013,Nasir2013,ParkClerckx2013,HuangLarsson2013}) and wireless powered communications (WPC) (see e.g. \cite{Lee2013,HuangLau,JuZhang2013,LiuZhangChua2014}), where the information transmission in the network is in the same or opposite direction of the WET, respectively. 

For both SWIPT and WPC, how to optimize the energy transfer efficiency from the energy transmitter (ET) to one or more energy receivers (ERs) by combating the severe signal power loss over distance is a challenging problem. To efficiently solve this problem, multi-antenna or multiple-input multiple-output (MIMO) techniques, which have been successfully applied in wireless communication systems to improve the information transmission rate and reliability over wireless channels, were also proposed for WET \cite{ZhangHo2013}. Specifically, deploying multiple antennas at the ET enables focusing the transmitted energy to destined ERs via beamforming, while equipping multiple antennas at each ER increases the effective aperture area, both leading to  improved end-to-end energy transfer efficiency. For the point-to-point MIMO WET system, it has been shown in \cite{ZhangHo2013} that energy beamforming is optimal to maximize the energy transfer efficiency by transmitting with only one single energy beam at the ET, which is in sharp contrast to  the celebrated spatial multiplexing technique used in the point-to-point MIMO communication system which applies multiple beams to maximize the information transmission rate \cite{Telatar1999}.

\begin{figure}
\centering
 \epsfxsize=1\linewidth
    \includegraphics[width=8.8cm]{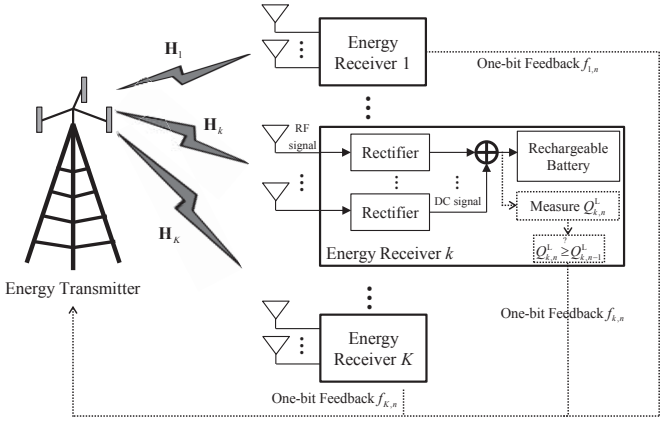}
\caption{A multiuser MIMO broadcast system for wireless energy transfer (WET).} \label{fig:system}
\end{figure}

In practical systems, the benefit of energy beamforming in MIMO WET crucially relies on the availability of the channel state information (CSI) at the ET. However, acquiring such CSI is particularly challenging in WET systems, since existing methods for channel learning in wireless communication  (see e.g. \cite{Love2008} and the references therein) may be no longer applicable. For example, one well-known solution to acquire the CSI at the transmitter in conventional wireless communication is by estimating the reverse link channel based on the training signals sent by the receiver. However, this method only applies to systems operating in time-division duplex (TDD) and critically depends on the accuracy of the assumption made on the reciprocity between the forward and reverse link channels.  Furthermore, applying this method to WET systems requires a more careful design of energy-efficient training signals at the ER, since they consume part of the ER's energy that is harvested from the ET. Alternatively, another commonly adopted solution to obtain CSI at the transmitter in wireless communication is by sending training signals from the transmitter to the receiver, through which the receiver can estimate the channel and then send the estimated channel back to the transmitter via a feedback channel. This method applies to both TDD and frequency-division duplex (FDD) based systems; however, it requires complex baseband signal processing at the receiver for channel estimation and feedback, which may not be implementable at the ER in WET system due to its practical hardware limitation.{\footnote{Fig. \ref{fig:system} shows a commonly used ER design for WET \cite{ZhouZhangHo2013}, in which each receive antenna (also known as rectenna) first converts the received RF signal to a direct current (DC) signal via a rectifier, and then the DC signals from all receive antennas are combined to charge a battery. Evidently, it is difficult in this ER design to incorporate baseband signal processing for channel estimation.}} To overcome the above drawbacks of existing methods, it is desirable to investigate new channel learning and feedback schemes for MIMO WET systems by taking into account the hardware limitation at each ER, which motivates this work.

In this paper, we consider a multiuser MIMO system for WET as shown in Fig. \ref{fig:system}, where one ET with $M_T > 1$ transmit antennas broadcasts wireless energy to a group of $K$ ERs each with $M_R\ge 1$ receive antennas via transmit energy beamforming over a given frequency band. We assume that the ERs can send their feedback information to the ET perfectly over orthogonal feedback channels (by e.g. piggybacking the feedback information with their uplink data in a wireless powered sensor network \cite{JuZhang2013}). Under this system setup, we propose a two-phase transmission protocol for channel learning and energy transmission, respectively. In the channel learning phase, the ET aims to learn the MIMO channels to different ERs by adjusting the training signals according to the individual feedback information from each ER. Based on the estimated channels, in the energy transmission phase, the ET then designs optimal transmit energy beamforming to maximize the weighted sum-power transferred to all ERs. In particular, we propose a new channel learning algorithm that requires only one feedback bit from each ER per feedback interval. Specifically, each feedback bit indicates the increase or decrease of the harvested energy at each ER in the present versus the previous intervals, which can be practically measured at each ER without changing the existing energy harvesting circuits as shown in Fig. \ref{fig:system} (by e.g. connecting an ``energy meter'' at the sum output of the DC signals from different receive antennas). Based on such feedback information, the ET adjusts its transmitted training signals in subsequent feedback intervals during the channel learning phase and at the same time obtains improved estimates of the MIMO channels to different ERs.

It is worth noting that there have been several alternative schemes reported in the literature for one-bit feedback based channel learning, e.g., {\it cyclic Jacobi technique (CJT)} \cite{Noam2013}, {\it gradient sign} \cite{BanisterZeidler2003}, and {\it distributed beamforming} \cite{Mudumbai2010}, which have been proposed and studied in different application scenarios. Specifically, the CJT algorithm was proposed for the secondary transmitter (ST) to learn its interference channel to the primary receiver (PR) in a MIMO cognitive radio system \cite{Noam2013}, in which the ST adjusts its transmitted signals over consecutive time slots based on the one-bit information indicating the increase or decrease of its resulted interference power at the PR, which is extracted from the feedback signals of the PR. The gradient sign algorithm was proposed to estimate the dominant eigenmode of a point-to-point MIMO channel \cite{BanisterZeidler2003}, in which the transmitter obtains a one-bit feedback from the receiver per time slot, which indicates the increase or decrease of the signal-to-noise-ratio (SNR) at the receiver over two consecutive slots. The distributed beamforming algorithm was proposed to learn the channel phases in a system consisting of multiple distributed single-antenna transmitters simultaneously sending a common message to a single-antenna receiver \cite{Mudumbai2010}, where each transmitter updates its own signal phase in a distributed manner by using the one-bit feedback from the receiver indicating whether the current SNR is larger or smaller than its recorded highest SNR so far. Notice that the above three algorithms can all be applied to one-bit feedback based channel learning in the MIMO WET system of our interest; however, these methods have the common limitation that they can only be used to learn the eigenvectors or the dominant eigenmode of a single-user MIMO channel matrix at each time, instead of learning multiple users' MIMO channels exactly at the same time. As a result, they may not achieve the optimal energy transfer efficiency in the multiuser MIMO WET system based on one-bit feedback.

In this paper, we propose a new approach to design the one-bit feedback based MIMO channel learning for WET by applying the celebrated analytic center cutting plane method (ACCPM) in convex optimization \cite{Boyd:ConvexII}. To the authors' best knowledge, this paper is the first attempt to apply the ACCPM approach for the design of channel learning with one-bit feedback. For our proposed ACCPM based channel learning algorithm, we first provide an analysis for its worst-case convergence. It is shown that the ACCPM based channel learning can obtain the estimates of all $K$ MIMO channels each with arbitrary number of receive antennas, $M_R$, in at most $\mathcal O\left(\left\lceil\frac{K}{M_T^2-1}\right\rceil\frac{M_T^3}{\varepsilon^2} \right)$ number of feedback intervals with $\varepsilon > 0$ denoting a desired accuracy, and $\lceil\cdot\rceil$ representing the ceiling function of real numbers. From this result, it is further inferred that when $K \le M_T^2-1$, the proposed algorithm has the same analytic convergence performance regardless of the number of ERs, $K$, which shows its benefit of simultaneously learning multiuser MIMO channels. Finally, we compare the performance of our proposed channel learning algorithm against the aforementioned three benchmark algorithms in terms of both convergence speed and energy transfer efficiency. It is shown through extensive simulations that our proposed algorithm achieves faster convergence for channel learning as well as higher energy transfer efficiency than the other three algorithms; while the performance gain of our proposed algorithm becomes more significant as the number of ERs in the WET system increases.

The remainder of this paper is organized as follows. Section \ref{SysMod} introduces the system model and the two-phase transmission protocol. Section \ref{sec:one-bit} presents the proposed channel learning algorithm with one-bit feedback for the point-to-point or single-user MIMO WET system as well as its convergence analysis. Section \ref{sec:one-bit:multi} extends the channel learning algorithm and analysis to the general multiuser WET system. Section \ref{sec:numerical} provides simulation results to evaluate the performance of our proposed algorithm as compared to other benchmark algorithms. Finally, Section \ref{sec:conclusion} concludes the paper.

{\it Notation:} Boldface letters refer to vectors (lower  case) or matrices (upper case). For a square matrix $\mv{S}$, $\det(\mv{S})$ and ${\mathtt{tr}}(\mv{S})$ denote its determinant and trace, respectively, while $\mv{S}\succeq \mv{0}$ and $\mv{S}\preceq \mv{0}$ mean that $\mv{S}$ is positive semi-definite and negative semi-definite, respectively. For an arbitrary-size matrix $\mv{M}$, $\|\mv{M}\|_{\rm F}$, ${\mathtt{rank}}(\mv{M})$, $\mv{M}^H$, and $\mv{M}^T$ denote the Frobenius norm, rank, conjugate transpose and transpose of $\mv{M}$, respectively. $\mv{I}$, $\mv{0}$, and $\mv 1$ denote an identity matrix, an all-zero matrix, and an all-one column vector, respectively, with appropriate dimensions. $\mathbb{C}^{x\times y}$ and $\mathbb{R}^{x\times y}$ denotes the space of $x\times y$ complex and real matrices, respectively. ${\mathbb{E}}(\cdot)$ denotes the statistical expectation. $\|\mv{x}\|$ denotes the Euclidean norm of a complex vector $\mv{x}$, and $|z|$ denotes the magnitude of a complex number $z$. $j$ denotes the complex number $\sqrt{-1}$.

\section{System Model}\label{SysMod}

We consider a multiuser MIMO broadcast system for WET as shown in Fig. \ref{fig:system}, where one ET with $M_T>1$ transmit antennas delivers wireless energy to a group of $K \ge 1$ ERs, denoted by the set $\mathcal{K}=\{1,\ldots,K\}$. For notational convenience, each ER is assumed to be deployed with the same number of $M_R\ge 1$ receive antennas, while our results directly apply to the case when each ER is with different number of antennas. We assume a quasi-static flat fading channel model, where the channel from the ET to each ER remains constant within each transmission block of our interest and may change from one block to another. We denote each block duration as $T$, which is assumed to be sufficiently long for typical low-mobility WET applications.

We consider linear transmit energy beamforming at the multiple-antenna ET. Without loss of generality, we assume that the ET sends $d \le M_T$ energy beams, where $d$ is our design parameter to be specified later. Let the $m$th beamforming vector be denoted by $\mv{w}_m\in \mathbb{C}^{M_T \times 1}$ and its carried energy-modulated signal by $s_m$, $m\in\{1,\ldots,d\}$. Then the transmitted signal at ET is given by $\mv{x} = \sum_{m=1}^{d} \mv{w}_m s_m.$ Since $s_m$'s do not carry any information, they can be assumed to be independent sequences from an arbitrary distribution with zero mean and unit variance, i.e., $\mathbb{E}\left(|s_m|^2\right)=1, \forall m$. Furthermore, we denote the transmit covariance matrix as  $\mv{S}=\mathbb{E}(\mv{x}\mv{x}^H)=\sum_{m=1}^{d}\mv{w}_m \mv{w}_m^H \succeq \mv{0}$. Note that given any positive semi-definite matrix $\mv{S}$, the corresponding energy beams $\mv{w}_1, \ldots, \mv{w}_d$ can be obtained from the eigenvalue decomposition (EVD) of  $\mv{S}$ with $d=\mathtt{rank}(\mv{S})$. Assume that the ET has a transmit sum-power constraint $P$ over all transmit antennas; then we have $\mathbb{E}(\|\mv{x}\|^2)= \sum_{m=1}^{d} \|\mv{w}_m\|^2 = \mathtt{tr}(\mv{S}) \le P$.

With transmit energy beamforming, each ER $k$ can harvest the wireless energy carried by all $d$ energy beams from its $M_R$ receive antennas. Denote $\mv{H}'_k \in \mathbb{C}^{M_R \times M_T}$ as the MIMO channel matrix from the ET to ER $k$, and $\mv{G}'_k \triangleq {\mv{H}_k'}^H\mv{H}'_k \succeq \mv{0}$. Then by letting ${\gamma_k}$ denote the Frobenius norm of the matrix ${\mv{G}}'_k$, i.e., $\gamma_k = \|{\mv{G}}'_k\|_{\rm F}$, we obtain the normalized channel matrix from the ET to ER $k$ as ${{\mv{H}}_k} \triangleq \mv{H}'_k/\sqrt{\gamma_k}$ (or ${\mv{G}}_k \triangleq \mv{G}'_k/\gamma_k$) with $\|{\mv{G}}_k\|_{\rm F} = \|{\mv{H}}_k^H{{\mv{H}}_k}\|_{\rm F} = 1, \forall k\in\mathcal{K}$. Accordingly, the harvested energy at ER $k$ over one block of interest is expressed as \cite{ZhangHo2013}
\begin{align}
Q_k = \varsigma T\mathbb{E}\left(\left\|\mv{H}_k'\mv{x}\right\|^2\right) = \varsigma T \gamma_k \mathtt{tr}(\mv{G}_k\mv{S}), k\in\mathcal{K},\label{eqn:1}
\end{align}
where $0 < \varsigma \le 1 $ denotes the energy harvesting efficiency at each receive antenna (cf. Fig. \ref{fig:system}). Since $\varsigma$ is a constant, we normalize it as $\varsigma = 1$ in the sequel of this paper unless otherwise specified. It is assumed that each ER $k$ cannot directly estimate the MIMO channel $\mv H'_k$ (or $\mv G'_k$) given its energy harvesting receiver structure (cf. Fig. \ref{fig:system}); instead, it can measure its average harvested power over a certain period of time by simply connecting an ``energy meter'' at the combined DC signal output shown in Fig. \ref{fig:system}.

We aim to design the energy beams at the ET to maximize the weighted sum-energy transferred to $K$ ERs, i.e., $Q\triangleq\sum_{k\in\mathcal{K}} \alpha_kQ_k$ with $Q_k$ given in (\ref{eqn:1}), over each transmission block subject to a given transmit sum-power constraint, where $\alpha_k  \ge 0$ denotes the energy weight for ER $k\in\mathcal{K}$ with $\sum_{k\in\mathcal{K}}\alpha_k = 1$. In order to ensure certain fairness among different ERs for WET, it is desirable to assign higher energy weights to the ERs more far apart from the ET. Accordingly, in this paper we set the energy weight to be proportional to the reciprocal of the channel power gain to the respective ER, i.e.,
\begin{align}\label{eqn:revision1}
\alpha_k = \frac{ 1/{\gamma_k}}{\sum_{l\in\mathcal{K}} (1/{\gamma_l})}, k\in\mathcal{K}.
\end{align}
As a result, the weighted sum-energy transferred to $K$ ERs can be re-expressed as $Q = T\gamma\mathtt{tr}(\mv{G}\mv{S})$ with $\mv{G} \triangleq \sum_{k\in\mathcal{K}} \mv{G}_k$ and $\gamma \triangleq \frac{ 1}{\sum_{l\in\mathcal{K}} (1/{\gamma_l})}$. As a result, we can formulate the weighted sum-energy maximization problem as
\begin{align}
\mathop\mathtt{max}\limits_{\mv{S}}~&T\gamma\mathtt{tr}(\mv{G}\mv{S})\nonumber\\
\mathtt{s.t.}~~&\mathtt{tr}(\mv{S})\le P,~\mv{S}\succeq \mv{0}.\label{eqn:problem:MaxK1}
\end{align}
It has been shown in \cite{ZhangHo2013} that the optimal solution  to (\ref{eqn:problem:MaxK1}) is given by $\mv{S}^* = {P}\mv{v}_E\mv{v}_E^H$, which achieves the maximum value of $Q_{\rm{max}}=T\gamma P\lambda_E$, with $\lambda_E$ and $\mv{v}_E$ denoting the dominant eigenvalue and its corresponding eigenvector of $\mv{G}$, respectively. Since $\mathtt{rank}\left(\mv{S}^*\right)=1$, this solution implies that sending one energy beam (i.e., $d=1$) in the form of $\mv{w}_1=\sqrt{P}\mv{v}_E$ is optimal for our multiuser MIMO WET system of interest. This solution is thus referred to as the {\it optimal energy beamforming (OEB)} for a given $\mv G$. Here, implementing the OEB only requires the ET to have the perfect knowledge of the $K$ normalized MIMO channels, $\mv{G}_1, \ldots, \mv{G}_K$, but does not require its knowledge of the average channel gain $\gamma_k$'s.{\footnote{Note that the OEB design here can be extended to other cases with different energy fairness considered, by modifying the energy weight $\alpha_k$'s (instead of setting them as in (\ref{eqn:revision1})). In such cases, it may be necessary for the ET to have an estimate of the average channel gain $\gamma_k$'s for setting $\alpha_k$'s. Since $\gamma_k$'s change slowly over time, they can be coarsely estimated in practice by e.g. measuring the received signal strength from each ER in the reverse link, by assuming a weaker form of channel reciprocity.}}


\begin{figure}
\centering
 \epsfxsize=1\linewidth
    \includegraphics[width=8.8cm]{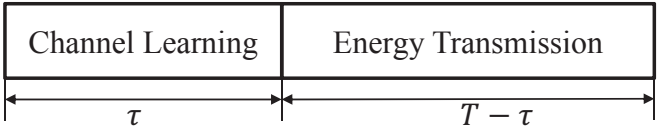}\vspace{0.3em}
\caption{The two-phase transmission protocol.} \label{fig:protocol}
\end{figure}

In order for the ET to practically estimate the MIMO channels, $\mv{G}_1, \ldots, \mv{G}_K$, we propose a transmission protocol for the multiuser MIMO WET system as shown in Fig. \ref{fig:protocol}, which consists of two consecutive phases in each transmission block for the main purposes of channel learning and energy transmission, respectively. We explain these two phases of each transmission block in more detail as follows.

The channel learning phase corresponds to the first $\tau$ amount of time in each block of duration $T$, which is further divided into $N_L$ feedback intervals each of length $T_s$, i.e., $\tau=N_LT_s$. For convenience, we assume that $N=T/T_s$ is an integer denoting the total block length in number of feedback intervals. During this phase, the ET transmits different training signals (each specified by a corresponding transmit covariance matrix) over $N_L$ feedback intervals. Let the transmit covariance at the ET in interval $n\in\{1,\ldots,N_L\}$ be denoted by $\mv{S}_n^{\rm{L}}$. Then the transferred energy to the $k$th ER over the $n$th interval is given by $Q_{k,n}^{\rm{L}}=T_s\gamma_k\mathtt{tr}(\mv{G}_k\mv{S}_n^{\rm{L}})$. In the meanwhile, ER $k$ measures its harvested energy amount $Q_{k,n}^{\rm{L}}$ and based on it feeds back one bit at the end of the $n$th interval, denoted by $f_{k,n} \in\{0,1\}$, to indicate whether the harvested energy in the $n$th interval is larger (i.e., $f_{k,n}=0$) or smaller (i.e., $f_{k,n}=1$) than that in the $(n-1)$th interval, $n=1,\ldots,N_L$. For the convenience of our analysis later, we set $f_{k,n}\leftarrow 2f_{k,n}-1$ such that $f_{k,n}\in\{-1,1\}$. More specifically, if $Q_{k,n}^{\rm{L}} \ge Q_{k,n-1}^{\rm{L}}$, then $f_{k,n}=-1$; while if $Q_{k,n}^{\rm{L}} < Q_{k,n-1}^{\rm{L}}$, then $f_{k,n}=1$. We also denote $Q_{k,0}^{\rm{L}} \triangleq 0$ and equivalently $\mv{S}_{0}^{\rm{L}} = \mv{0}$ for convenience. Notice that the feedback interval $T_s$ should be designed considering the practical  feedback link rate from each ER to the ET as well as the sensitivity of the energy meter at each ER. For the purpose of exposition, we assume in this paper that $Q_{k,n}^{\rm{L}}$'s are all perfectly measured at corresponding ERs, and thus $f_{k,n}$'s are all accurately determined at the ERs and then sent back to the ET without any error.{\footnote{In practice, there exist measurement errors for estimating $Q_{k,n}^{\rm{L}}$'s at the ERs due to the rectifier noise and feedback errors in the received $f_{k,n}$'s at the ET due to the imperfect reverse links from the ERs, both of which result in inaccurate $f_{k,n}$'s at the ET. It is thus interesting to investigate their effects on the performance of our proposed channel learning algorithm with one-bit feedback, which, however, are beyond the scope of this paper.}} Furthermore, we assume that the consumed energy for sending the one-bit feedback $f_{k,n}$'s is negligible at each ER as compared to its average harvested energy. At the end of the channel learning phase, by using the collected feedback bits $\{f_{k,n}\}$ from all ERs, the ET can obtain an estimate of the normalized MIMO channel $\mv{G}_k$ for each ER $k$, which is denoted by $\tilde{\mv{G}}_k$, $k\in\mathcal{K}$. The details of training signal design and channel estimation at the ET based on the one-bit feedback information from one or more ERs will be given later in Sections \ref{sec:one-bit} and \ref{sec:one-bit:multi}.

The subsequent energy transmission phase in each block corresponds to the remaining $T-\tau$ amount of time. Given the estimated $\tilde{\mv{G}}_k$'s from the channel learning phase, we can obtain the estimate of $\mv{G}$ as $\tilde{\mv{G}} = \sum_{k\in\mathcal{K}} \tilde{\mv{G}}_k$, and accordingly have the estimate of its dominant eigenvector $\mv{v}_E$ as $\tilde{\mv{v}}_E$. Then based on the principle of OEB, the ET sets the (rank-one) transmit covariance in the energy transmission phase as $\mv{S}^{\rm{E}} = P\tilde{\mv{v}}_E\tilde{\mv{v}}_E^H$. Accordingly, the weighted sum-energy transferred to all ERs during this phase is expressed as $Q^{\rm{E}} = (T-\tau)P\gamma\tilde{\mv{v}}_E^H{\mv{G}}\tilde{\mv{v}}_E$.

Combining the above two phases, the total weighted sum-energy transferred to all the $K$ ERs over one particular block is given by
\begin{align}\label{eqn:Qpro}
Q_{\rm{total}}=\sum_{n=1}^{N_L}T_s\gamma\mathtt{tr}(\mv{G}\mv{S}_n^{\rm{L}}) + (T-\tau)P\gamma\tilde{\mv{v}}_E^H{\mv{G}}\tilde{\mv{v}}_E.
\end{align}
In (\ref{eqn:Qpro}), we observe that if the estimated MIMO channel $\tilde{\mv{G}}_k$'s are all accurate with a given finite $N_L$ (or $\tau$), then it follows that $\tilde{\mv{v}}_E^H{\mv{G}}\tilde{\mv{v}}_E\approx \lambda_E$. In this case, we can have $Q_{\rm{total}} \to Q_{\rm{max}}$ by increasing the block duration, i.e., $N\to\infty$ or $T\to\infty$. However, given finite $N$ or $T$ (which needs to be chosen to be smaller than the channel coherence time in practice), there is in general a trade-off in setting the time allocations, i.e., $N_L$ versus $N-N_L$,  between the channel learning and energy transmission phases in order to maximize $Q_{\rm{total}}$ in (\ref{eqn:Qpro}), as will be demonstrated latter by our numerical results in Section \ref{sec:numerical}.

In the above proposed two-phase transmission protocol for multiuser MIMO WET, the key challenge lies in the design of channel learning algorithms at the ET to estimate the normalized MIMO channel $\mv{G}_k$'s based only on the one-bit feedbacks from different ERs in the first channel learning phase, which is thus our focus of study in the rest of this paper. In the next two sections, we first present the channel learning algorithm for the special case of one single ER to draw useful insights, and then extend the algorithm to the general case with multiple ERs.

\section{Channel Learning with One-Bit Feedback: Single-User Case}\label{sec:one-bit}

In this section, we consider the point-to-point or single-user MIMO WET system with $K=1$ ER. For notational convenience, we remove the user subscript $k$ in this case, and thus denote the harvested energy amount and the feedback bit at each interval $n$ as $Q^{\rm L}_n$ and $f_n$, respectively, $n\in\{1,\ldots,N_L\}$. Furthermore, since there is only one ER in the system, we denote its channel power gain as $\gamma = \gamma_1$, and its normalized channel matrix to be estimated as $\mv{G} = \mv{G}_1$.

We aim to propose a new channel learning algorithm for the ET to estimate the single-user MIMO channel $\mv{G}$ based on the one-bit feedbacks from the ER over training intervals in the channel learning phase. The proposed algorithm is based on the celebrated ACCPM in convex optimization \cite{Boyd:ConvexII}. In the following, we first introduce ACCPM,\footnote{We refer the readers to  \cite{Boyd:ConvexII} for more details of ACCPM.} then present the ACCPM based channel learning algorithm with one-bit feedback, and finally provide its convergence analysis.

\subsection{Introduction of ACCPM}
ACCPM is an efficient localization and cutting plane method for solving general convex or quasi-convex optimization problems \cite{Boyd:ConvexII,SunTohZhao2002}, with the goal of finding one feasible point in a convex target set $\mathcal{X}\subseteq \mathbb{R}^{m \times 1}, m\ge 1$, where $\mathcal{X}$ can be the set of optimal solutions to the optimization problem. Suppose that any point in the target set $\mathcal{X}$ is known {\it a priori} to be contained in a convex set $\mathcal{P}_0$, i.e., $\mathcal{X} \subseteq \mathcal{P}_0$. $\mathcal{P}_0$ is referred to as the initial working set. The basic idea of ACCPM is to query an {\it oracle} for localizing the target set $\mathcal{X}$ through finding a sequence of convex working sets, denoted by $\mathcal{P}_1,\cdots,\mathcal{P}_i, \cdots.$ At each iteration $i \ge 1$, we query the oracle at a point ${\mv{x}}^{(i)}\in \mathbb{R}^{m \times 1}$, where ${\mv{x}}^{(i)}$ is chosen as the analytic center of the previous working set $\mathcal{P}_{i-1}$. If ${\mv{x}}^{(i)}\in\mathcal{X}$, then the algorithm ends. Otherwise, the oracle returns a {\it cutting plane}, i.e., $\mv{a}_i\neq \mv{0}$ and $b_i$ satisfying that
\begin{align}\label{eqn:APCCM:1}
\mv{a}_i^T\mv{z} \le b_i~{\rm for}~\mv{z} \in \mathcal{X},
\end{align}
which indicates that $\mathcal{X}$ should lie in the half space of $\mathcal{H}_i = \{\mv{z}|\mv{a}_i^T\mv{z} \le b_i\}$. After the querying, the working set is then updated as $\mathcal{P}_i = \mathcal{P}_{i-1}\cap\mathcal{H}_i$. By properly choosing the cutting plane in (\ref{eqn:APCCM:1}) based on ${\mv{x}}^{(i)}$, we can have $\mathcal{P}_0\supseteq\cdots\supseteq\mathcal{P}_i\supseteq\mathcal{X}$. Therefore, the returned working set $\mathcal{P}_i$ will be reduced and eventually approach the target set $\mathcal{X}$ as $i\to\infty$.

It is worth noting that given query point ${\mv{x}}^{(i)}$, if the cutting plane $\mv{a}_i^T\mv{z} = b_i$ in (\ref{eqn:APCCM:1}) contains ${\mv{x}}^{(i)}$, then it is referred to as a {\it neutral cutting plane}; if $\mv{a}_i^T{\mv{x}}^{(i)}  > b_i$, i.e., ${\mv{x}}^{(i)}$ lies in the interior of the cut half space, then it is named a {\it deep cutting plane}; otherwise, it is called as a {\it shallow cutting plane}. For ACCPM, a deep or at least neutral cutting plane is required in each iteration.

\subsection{ACCPM Based Single-User Channel Learning}\label{channelLearning:K1}

In this subsection, we present the proposed channel learning algorithm based on ACCPM. First, we define the target set for our problem of interest. Recall that our goal is to obtain an estimate of the normalized channel matrix $\mv{G}$, which is equivalent to finding any positively scaled estimate of ${\mv{G}}$. As a result, we define the target set as $\mathcal{X}=\{\bar{\mv{G}}|\mv{0}\preceq \bar{\mv{G}}\preceq \mv{I},~\bar{\mv{G}} = \beta\mv{G},\forall \beta>0\}$, which contains all scaled matrices of $\mv{G}$ satisfying that $\mv{0}\preceq \bar{\mv{G}}\preceq \mv{I}$. Since $\mv{0}\preceq \bar{\mv{G}}\preceq \mv{I}$ is known {\it a priori}, we have the initial convex working set as $\mathcal{P}_0 = \{\bar{\mv{G}}|\mv{0}\preceq \bar{\mv{G}}\preceq \mv{I}\},$ i.e., $\mathcal{X} \subseteq \mathcal{P}_0$.

Next, we show that the one-bit feedback $f_n$'s in the $N_L$ feedback intervals play the role of oracle in ACCPM for our problem, which return a sequence of working sets $\{\mathcal{P}_n\}$ to help localize the target set $\mathcal{X}$. Consider each feedback interval as one iteration. Then, for any feedback interval $n\in\{2,\ldots,N_L\}$,{\footnote{Note that for interval $n=1$, it always holds that $Q_1^{\rm{L}} \ge Q_0^{\rm{L}} = 0$, and thus the one-bit feedback information is always $f_{1} = -1$, which does not contain any useful information for localizing the target set $\mathcal{X}$.}} by querying the one-bit feedback $f_n$, the ET can obtain the following inequality for $Q_n^{\rm{L}}$ and $Q_{n-1}^{\rm{L}}$ (recall that $Q_n^{\rm{L}}=T_s\gamma\mathtt{tr}(\mv{G}\mv{S}_n^{\rm{L}})$):
\begin{align}
f_n\mathtt{tr}\left(\mv{G} (\mv{S}_n^{\rm{L}}-\mv{S}_{n-1}^{\rm{L}})\right) \le 0,\label{eqn:energybeam:0}
\end{align}
which can be regarded as a cutting plane such that  $\mv{G}$ lies in the half space of  $\mathcal{H}_n=\{\bar{\mv{G}}|f_n\mathtt{tr}\left(\bar{\mv{G}} (\mv{S}_n^{\rm{L}}-\mv{S}_{n-1}^{\rm{L}})\right)  \le 0\}$. Accordingly, by denoting $\mathcal{P}_1=\mathcal{P}_0$, we can obtain the working set $\mathcal{P}_n$ at interval $n\ge 2$ by updating $\mathcal{P}_n = \mathcal{P}_{n-1}\cap\mathcal{H}_n$, or equivalently,
\begin{align}
\mathcal{P}_n =  &\bigg\{\bar{\mv{G}}\big|\mv{0}\preceq\bar{\mv{G}}\preceq \mv{I},~f_i\mathtt{tr}\left(\bar{\mv{G}}\left(\mv{S}_i^{\rm{L}}-\mv{S}_{i-1}^{\rm{L}}\right)\right) \le 0, \nonumber \\ &~~~~~~~~~~~~~~~~~~~~~~~~~~~~~~~~~~~~~~~~~~~~2\le i\le n\bigg\}.\label{learning:4}
\end{align}It is evident that $\mathcal{P}_0=\mathcal{P}_1\supseteq\mathcal{P}_2\supseteq\cdots\supseteq\mathcal{P}_{N_L}\supseteq\mathcal{X}$.

From (\ref{learning:4}), we can obtain the analytic center of $\mathcal{P}_{n}$, denoted as $\tilde{\mv{G}}^{(n)}$, which is explicitly given by \cite{SunTohZhao2002}{\footnote{Since the matrix to be estimated (i.e., $\mv{G}$) is complex, we use $2\log \det \left(\bar{\mv{G}}\right)$ and $2\log \det \left(\mv{I}-\bar{\mv{G}}\right)$ in (\ref{eqn:13}) to compute the analytic centers, instead of  $\log \det \left(\bar{\mv{G}}\right)$ and $\log \det \left(\mv{I}-\bar{\mv{G}}\right)$ as used in \cite{SunTohZhao2002} for the case of real matrices. Our new definition in (\ref{eqn:13}) will facilitate the convergence proof for the proposed algorithm (see Appendix \ref{appendix:1}).}}
\begin{align}
 \tilde{\mv{G}}^{(n)} =&\mathtt{arg}\mathop\mathtt{min}_{\mv{0}\preceq\bar{\mv{G}}\preceq \mv{I}}~- 2\log \det \left(\bar{\mv{G}}\right) -  2\log \det \left(\mv{I}-\bar{\mv{G}}\right) \nonumber \\
& -\sum_{i=2}^{n} \log\left(-f_i\mathtt{tr}\left(\bar{\mv{G}} \left(\mv{S}_i^{\rm{L}}-\mv{S}^{\rm{L}}_{i-1}\right)\right)\right), n\ge 0.\label{eqn:13}
\end{align}
Since the problem in (\ref{eqn:13}) can be shown to be convex \cite{Boyd}, it can be solved by standard convex optimization techniques, e.g., CVX \cite{cvx}. Notice that $\tilde{\mv{G}}^{(n)}$ is also the query point for the next feedback interval $n+1$.

Up to now, we have obtained the query point at each interval $n$, $\tilde{\mv{G}}^{(n-1)}$, and the cutting plane given by (\ref{eqn:energybeam:0}) for ACCPM. To complete our algorithm, we also need to ensure that the resulting cutting plane is at least neutral given $\tilde{\mv{G}}^{(n-1)}$. This is equivalent to constructing the transmit covariance $\mv{S}_n^{\rm{L}}$'s such that
\begin{align}\mathtt{tr}\left(\tilde{\mv{G}}^{(n-1)} \left(\mv{S}_n^{\rm{L}}-\mv{S}_{n-1}^{\rm{L}}\right)\right) = 0, n=2,\ldots,N_L.\label{eqn:neutral}\end{align} We find such $\mv{S}^{\rm{L}}_n$'s by setting $\mv{S}^{\rm{L}}_1=\frac{P}{M_T}\mv{I}$ for interval $n=1$ and
\begin{align}\label{eqn:A}
\mv{S}^{\rm{L}}_n=\mv{S}^{\rm{L}}_{n-1}+\mv{B}_n
\end{align}for the remaining intervals $n =2,\ldots,N_L$, where $\mv{B}_n\in\mathbb{C}^{M_T\times M_T}$ is a Hermitian probing matrix that is neither positive nor negative semi-definite in general. With the above choice, finding a pair of $\mv{S}^{\rm{L}}_n$ and $\mv{S}^{\rm{L}}_{n-1}$ to satisfy (\ref{eqn:neutral}) is simplified to finding the probing matrix $\mv{B}_n$ satisfying $\mathtt{tr}(\tilde{\mv{G}}^{(n-1)}{\mv{B}}_n)=0,n =2,\ldots,N_L.$ To find such ${\mv{B}}_n$ for the $n$th interval, we define a vector operation $\mathrm{cvec}(\cdot)$ that maps a complex Hermitian matrix $\mv{X} \in \mathbb{C}^{m \times m}$ to a real vector $\mathrm{cvec}(\mv{X}) \in \mathbb{R}^{m^2 \times 1}, m\ge1$, where all elements of $\mathrm{cvec}(\mv{X})$ are independent from each other, and $\mathtt{tr}({\mv{X}}{\mv{Y}}) = (\mathrm{cvec}(\mv{X}))^T\mathrm{cvec}(\mv{Y})$ for any given complex Hermitian matrix $\mv{Y}$.\footnote{The mapping between the complex Hermitian matrix $\mv{X} \in \mathbb{C}^{m \times m}$ and the real vector $\mathrm{cvec}(\mv{X})  \in \mathbb{R}^{m^2 \times 1}$, $m \ge 1$, can be realized as follows. The first $m$ elements of $\mathrm{cvec}(\mv{X}) $ consist of the diagonal elements of $\mv{X}$ (that are real), i.e., $[\mv{X}]_{aa}$'s, $\forall a\in\{1,\ldots,m\}$, where $[\mv{X}]_{ab}$ denotes the element in the $a$th row and $b$th column of $\mv{X}$; the next $\frac{m^2-m}{2}$ elements of $\mathrm{cvec}(\mv{X}) $ are composed of the (scaled) real part of the upper (or lower) off-diagonal elements of $\mv{X}$, i.e., $\frac{[\mv{X}]_{ab} + [\mv{X}]_{ba}}{\sqrt{2}}$'s, $\forall a,b\in\{1,\ldots,m\}, a < b$; and the last $\frac{m^2-m}{2}$ elements of $\mathrm{cvec}(\mv{X}) $ correspond to the (scaled) imaginary part of the lower off-diagonal elements of $\mv{X}$, i.e.,  $j\frac{[\mv{X}]_{ab} - [\mv{X}]_{ba}}{\sqrt{2}}$'s, $\forall a,b\in\{1,\ldots,m\}, a < b$.} Accordingly, we can express $\tilde{\mv{g}}^{(n-1)} = \mathrm{cvec}\left(\tilde{\mv{G}}^{(n-1)}\right)$ and ${\mv{b}}_n=\mathrm{cvec}\left({\mv{B}}_n\right)$, where $\tilde{\mv{g}}^{(n-1)T}{\mv{b}}_n=\mathtt{tr}(\tilde{\mv{G}}^{(n-1)}{\mv{B}}_n)=0.
$
Due to the one-to-one mapping of $\mathrm{cvec}\left(\cdot\right)$, finding ${\mv{B}}_n$ is equivalent to finding ${\mv{b}}_n$ that is orthogonal to $\tilde{\mv{g}}^{(n-1)}$. Define a projection matrix $\mv{F}_n = \mv{I} - \frac{\tilde{\mv{g}}^{(n-1)}\tilde{\mv{g}}^{(n-1)T}}{\|\tilde{\mv{g}}^{(n-1)}\|^2}$. Then we can express  ${\mv{F}}_n = {\mv{V}}_n{\mv{V}}^T_n$, where ${\mv{V}}_n \in \mathbb{R}^{M_T^2\times (M_T^2-1)}$ satisfies ${\mv{V}}^T_n\tilde{\mv{g}}^{(n-1)}= \mv{0}$ and ${\mv{V}}^T_n{\mv{V}}_n = \mv{I}$. Thus, $\mv{b}_n$ can be any vector in the subspace spanned by ${\mv{V}}_n$. Specifically, we set
\begin{align}\label{learning:3}
\mv{b}_n = {\mv{V}}_n\mv{p},
\end{align}
where $\mv{p}\in\mathbb{R}^{(M_T^2-1)\times 1}$ is a randomly generated vector in order to make $\mv{b}_n$ independently drawn from the subspace. With the obtained $\mv{b}_n$, we have the probing matrix $\mv{B}_n = \mathrm{cmat}(\mv{b}_n)$,{\footnote{Note that $\mv{B}_n$ in general contains both positive and negative eigenvalues. As a result, the update in (\ref{eqn:A}) may not necessarily yield an $\mv{S}_n^{\rm L}$ that satisfies both $\mathtt{tr}(\mv{S}_n^{\rm L})\leq P$ and $\mv{S}_n^{\rm L}\succeq \mv{0}$. Nevertheless, by setting $\|\mv{p}\|$ to be sufficiently smaller than $P$, we can always find a $\mv{p}$ and its resulting $\mv{S}_n^{\rm L}$ satisfying the above two conditions with only a few random trials. In this paper, we choose  $\|\mv{p}\|=P/10$.}} where $\mathrm{cmat}(\cdot)$ denotes the inverse operation of $\mathrm{cvec}(\cdot)$. Accordingly, $\mv{S}_n^{\rm{L}}$ that satisfies the neutral cutting plane in (\ref{eqn:neutral}) is obtained.

To summarize, we present the ACCPM based channel learning algorithm with one-bit feedback for the single-user case in Table I as Algorithm 1. Note that in step 3) of the algorithm, the iteration terminates after $N_L$ feedback intervals of the channel leaning phase, and in step 4), the estimate of ${\mv{G}}$ is set as the normalized matrix of the analytic center of  $\mathcal{P}_{N_L}$, given by $\tilde{\mv{G}}=\frac{\tilde{\mv{G}}^{(N_L)}}{\left\|\tilde{\mv{G}}^{(N_L)}\right\|_{\rm F}}$. Accordingly, we can use the dominant eigenvector of $\tilde{\mv{G}}$ as the corresponding OEB $\tilde{\mv{v}}_E$ for the energy transmission phase in the single-user MIMO WET system.

\begin{table}[!t]\scriptsize
\caption{ACCPM Based Channel Learning Algorithm for Single-User Case}
\label{table2} \centering
\begin{tabular}{|p{8.5cm}|}
\hline
\textbf{Algorithm 1}\\
\hline\vspace{0.01cm}
  1) {\bf Initialization:} Set $n=0$, $Q_0^{\rm{L}}=0$, and $\mv{S}_1^{\rm{L}}=\frac{P}{M_T}\mv{I}$. \\
2) {\bf Repeat:}
  \begin{itemize} \setlength{\itemsep}{0pt}
    \item[a)] $n \gets n+1$;
    \item[b)] The ET transmits with $\mv{S}_n^{\rm{L}}$;
    \item[c)] The ER feeds back $f_n=-1$ (or $1$) if $Q_n^{\rm{L}} \ge Q_{n-1}^{\rm{L}}$ (or otherwise);
    \item[d)] The ET computes the query point $\tilde{\mv{G}}^{(n)}$ given in (\ref{eqn:13});
    \item[e)] The ET computes $\mv{b}_{n+1}$ from (\ref{learning:3}), obtains $\mv{B}_{n+1}=\mathrm{cmat}(\mv{b}_{n+1})$, and updates $\mv{S}^{\rm{L}}_{n+1}=\mv{S}_{n}^{\rm{L}}+\mv{B}_{n+1}$.
  \end{itemize}
  3) {\bf Until}  $n\ge N_L$.\\
  4) The ET estimates $\tilde{\mv{G}}={\tilde{\mv{G}}^{(N_L)}}\big/{\left\|\tilde{\mv{G}}^{(N_L)}\right\|_{\rm F}}$.\\
 \hline
\end{tabular}\vspace{-1em}
\end{table}

\subsection{Convergence Analysis: Single-User Case}\label{channelLearning:convergence}

For the ACCPM based channel learning algorithm given in Table \ref{table2}, we proceed to analyze its convergence performance by assuming that $N_L$ (and hence $N$) can be set to be arbitrarily large. We first have the following proposition.
\begin{proposition}\label{proposition:1}
Suppose that the target set $\mathcal{X}$ admits certain estimation errors specified by the desired accuracy $\varepsilon > 0$, i.e., $\mathcal{X}=\{\bar{\mv{G}}^\varepsilon |\mv{0}\preceq \bar{\mv{G}}^\varepsilon \preceq \mv{I},~\|\bar{\mv{G}}^\varepsilon-\beta\mv{G}\|_{\mathrm{F}}\le \varepsilon,\forall \beta>0\}$. Then the updated $\tilde{\mv{G}}^{(n)}$'s in the ACCPM based single-user channel learning algorithm will converge to a point in the target set $\mathcal{X}$ with $\|\tilde{\mv{G}}^{(n)}\|_{\rm F} \ge 1/4$ once the iteration index $n$ ($n\ge 1$) satisfies the following inequality:
\begin{align}
\varepsilon^2 > \frac{M_T+4M_T^2(2M_T+1)\log\left(1+\frac{n-1}{16M_T^2(2M_T+1)}\right)}{(4n+16M_T-4)\exp(\frac{2(n-1)c}{n+4M_T-1})},\label{eqn:convergence}
\end{align}
where $c>0$ is a constant, and the right-hand side in (\ref{eqn:convergence}) is monotonically decreasing with $n \ge 1$.
\end{proposition}
\begin{proof}
Note that for the ACCPM based channel learning in Algorithm 1, each iteration of $n>1$ returns one neutral cutting plane; as a result, the required iteration number in Proposition \ref{proposition:1} is equivalent to the total number of required neutral cutting planes. Based on this observation, Proposition \ref{proposition:1} can be proved by borrowing the convergence analysis results of the ACCPM for semi-definite feasibility problems in \cite{SunTohZhao2002}, which shows the worst-case complexity on the total number of required neutral cutting planes given certain solution accuracy. However, \cite{SunTohZhao2002} only considers the case with real matrices, while our ACCPM based channel learning algorithm corresponds to the case involving complex matrices. To overcome this issue, we first find an equivalent real counterpart for the complex ACCPM based channel learning in Algorithm 1, and then prove Proposition \ref{proposition:1} by showing the convergence behavior of the real counterpart algorithm based on the results in \cite{SunTohZhao2002}. The detailed proof is provided in Appendix \ref{appendix:1}.
\end{proof}

In Proposition \ref{proposition:1}, we have obtained the number of feedback intervals required for $\tilde{\mv{G}}^{(n)}$ to converge in the target set $\mathcal{X}$ subject to certain estimation errors, where $\tilde{\mv{G}}^{(n)}$ can be an estimate of any scaled matrix of $\mv{G}$ with $\|\tilde{\mv{G}}^{(n)}\|_{\rm F} \ge 1/4$. However, since our main objective is to estimate the normalized channel matrix $\mv{G}$, it is desirable to further provide the explicit number of required feedback intervals for $\tilde{\mv{G}}$ (the estimate of $\mv{G}$) to converge. This is shown in the following proposition based on Proposition \ref{proposition:1}.
\begin{proposition}\label{proposition:2}
The ACCPM based single-user channel learning algorithm obtains an estimate $\tilde{\mv{G}}$ for the normalized channel matrix $\mv G$ with $\|\tilde{\mv{G}}-{\mv{G}}\|_{\rm F}\le \varepsilon$ in at most $\mathcal O\left(\frac{M_T^3}{\varepsilon^2} \right)$ number of feedback intervals.
\end{proposition}
\begin{proof}
See Appendix \ref{appendix:proof2}.
\end{proof}

From Proposition \ref{proposition:2}, it is evident that the analytic convergence speed is only related to the number of transmit antennas, $M_T$, but does not depend on the number of receive antennas, $M_R$. This is intuitive, since our algorithm aims to learn the composite channel matrix of $\mv{G} = \mv{H}^H\mv H$, which is of size $M_T \times M_T$. It is worth pointing out that the result in Proposition \ref{proposition:2} provides merely a worst-case upper bound for the required number of feedback intervals, $N_L$; practically, the proposed algorithm can achieve the desired accuracy with much smaller number of feedback intervals, $N_L$, as will be shown by our numerical results in Section \ref{sec:numerical}.

\section{Channel Learning with One-Bit Feedback: Multiuser Case}\label{sec:one-bit:multi}

In this section, we extend the ACCPM based single-user channel learning algorithm to the general multiuser MIMO WET system with $K>1$ ERs. In the following, we first present the multiuser modification of the ACCPM based channel learning algorithm with one-bit feedback, and then provide its convergence analysis.

\subsection{ACCPM Based Multiuser Channel Learning}

In the multiuser case, we aim to implement ACCPM to learn the $K$ normalized channel matrices from the ET to all ERs, i.e., ${\mv{G}}_1, \ldots, {\mv{G}}_K$, by using the collected one-bit feedback information from them. To this end, we need to define the corresponding target set, working sets and query points for each ER $k$, and also find a set of {\it neutral} cutting planes for all $K$ ERs at each feedback interval.


For each ER $k\in\mathcal{K}$, similar to the single-user case, we define the target set as $\mathcal{X}_k=\{\bar{\mv{G}}_k|\mv{0}\preceq \bar{\mv{G}}_k\preceq \mv{I},~\bar{\mv{G}}_k = \beta\mv{G}_k,\forall \beta>0\}$, and have the working sets as
\begin{align}
\mathcal{P}_{k,n} =  &\big\{\bar{\mv{G}}_k\big|\mv{0}\preceq\bar{\mv{G}}_k\preceq \mv{I},~f_{k,i}\mathtt{tr}\left(\bar{\mv{G}}_k\left(\mv{S}_i^{\rm{L}}-\mv{S}_{i-1}^{\rm{L}}\right)\right) \le 0, \big. \nonumber\\
&\big.~~~~~~~~~~~~~~2\le i\le n\big\}, n \ge 0,\label{learning:4:multiK}
\end{align}
where the inequality of
\begin{align}
f_{k,n}\mathtt{tr}\left(\bar{\mv{G}}_{k} (\mv{S}_n^{\rm{L}}-\mv{S}_{n-1}^{\rm{L}})\right) \le 0\label{eqn:energybeam:0:multiK}
\end{align}
corresponds to a cutting plane obtained at the $n$th interval based on ER $k$'s feedback of $f_{k,n}$, $n=2,\ldots,N_L$. From (\ref{learning:4:multiK}), we can obtain the analytic center of $\mathcal{P}_{k,n}$ (also the query point for the next interval $n + 1$), given by
\begin{align}
 \tilde{\mv{G}}_k^{(n)} =&\mathtt{arg}\mathop\mathtt{min}_{\mv{0}\preceq\bar{\mv{G}}_k\preceq \mv{I}}- 2\log \det \left(\bar{\mv{G}}_k\right) -  2\log \det \left(\mv{I}-\bar{\mv{G}}_k\right) \nonumber\\& -\sum_{i=2}^{n} \log\left(-f_{k,i}\mathtt{tr}\left(\bar{\mv{G}}_k \left(\mv{S}_i^{\rm{L}}-\mv{S}^{\rm{L}}_{i-1}\right)\right)\right), n\ge 0.\label{eqn:13:multiK}
\end{align}
Thus, we have obtained the target set, working sets and query points for each ER $k \in \mathcal{K}$.

Now, to complete ACCPM, we also need to design the transmit covariance $\mv{S}_n^{\rm{L}}$'s to ensure that the cutting plane in (\ref{eqn:energybeam:0:multiK}) is neutral. That is, at interval $n=2,\ldots,N_L$, it is desirable for each ER $k\in\mathcal{K}$ that
\begin{align}
\mathtt{tr}\left(\tilde{\mv{G}}_k^{(n-1)} \left(\mv{S}_n^{\rm{L}}-\mv{S}_{n-1}^{\rm{L}}\right)\right) = 0.\label{eqn:neutral:K}
\end{align}
Note that given $K>1$ ERs in the system, the transmit covariance matrix $\mv{S}_n^{\rm{L}}$ needs to satisfy $K$ equations in (\ref{eqn:neutral:K}) for $k=1,\ldots,K,$ at the same time, in contrast to one single equation in the single-user case with $K=1$. If $K > M_T^2-1$, finding such an $\mv{S}_n^{\rm{L}}$ becomes infeasible, since in this case, (\ref{eqn:neutral:K}) corresponds to a set of $K$ equations with $M_T^2$ (real) unknowns. To overcome this issue, we propose to group the $K$ ERs into one or more subsets each consisting of no more than $M_T^2-1$ number of ERs; accordingly, at each feedback interval, the ET only needs to ensure that the updated transmit covariance $\mv{S}_n^{\rm{L}}$ satisfies (\ref{eqn:neutral:K}) for the ERs in the corresponding subset, instead of all $K$ ERs if $K > M^T_2-1$.

Specifically, we divide the $K$ ERs into $A = \lceil\frac{K}{M_T^2-1}\rceil$ subsets denoted by $\mathcal{K}_a = \left\{(a-1)\lceil\frac{K}{A}\rceil+1,\ldots,a\lceil\frac{K}{A}\rceil\right\}$, $\forall a\in\{1,\ldots,A-1\}$, and $\mathcal{K}_A = \left\{(A-1)\lceil\frac{K}{A}\rceil+1,\ldots,K\right\}$, where $\bigcup \limits_{a=1}^A \mathcal{K}_a = \mathcal{K}$ and $|\mathcal{K}_a| \le M_T^2-1, \forall a$. Accordingly, we also partition the $N_L$ feedback intervals in the channel learning phase of the two-phase protocol into $A$ subsets as shown in Fig. \ref{fig:protocol:multiuser}, which are given by $\mathcal{N}_a^{\rm L} = \left\{(a-1)\lceil\frac{N_L}{A}\rceil+1,\ldots,a\lceil\frac{N_L}{A}\rceil\right\}$, $\forall a\in\{1,\ldots,A-1\}$, and $\mathcal{N}_A^{\rm L} = \left\{(A-1)\lceil\frac{N_L}{A}\rceil+1,\ldots,N_L\right\}$, where $\bigcup \limits_{a=1}^A \mathcal{N}^{\rm L}_a = \{1,\ldots,N_L\}$ and $\mathcal{N}^{\rm L}_a \bigcap \mathcal{N}^{\rm L}_b = \phi, \forall a\neq b$.{\footnote{How to optimally group ERs and partition feedback intervals over different groups to achieve the best channel learning performance is an interesting problem, which, however, is beyond the scope of this paper.}} Notice that for each partitioned subset of feedback intervals, $\mathcal{N}_a^L$, only ERs in the corresponding user subset $\mathcal{K}_a$ need to send their one-bit feedbacks to the ET for learning their MIMO channels; accordingly, over the intervals in $\mathcal{N}_a^L$, the ET obtains cutting planes in (\ref{eqn:energybeam:0:multiK}) only for the corresponding ERs in $\mathcal{K}_a$. Therefore, based on the above partitions, if $K > M_T^2-1$, we need to slightly modify the working sets in (\ref{learning:4:multiK}) and the analytic centers (query points) in (\ref{eqn:13:multiK}) for each $k\in\mathcal{K}_a, a\in\{1,\ldots,A-1\}$ as
\begin{align}
\mathcal{P}_{k,n} = & \left\{\bar{\mv{G}}_k\big|\mv{0}\preceq\bar{\mv{G}}_k\preceq \mv{I},~f_{k,i}\mathtt{tr}\left(\bar{\mv{G}}_k\left(\mv{S}_i^{\rm{L}}-\mv{S}_{i-1}^{\rm{L}}\right)\right) \le 0, \right.\nonumber \\ &~~~~~~~~~~~~~~~~~~\left.i \in \{2,\ldots,n\} \bigcap \mathcal{N}_a^L\right\}, \label{learning:4:multiK:modi} \\
 \tilde{\mv{G}}_k^{(n)} =&\mathtt{arg}\mathop\mathtt{min}_{\mv{0}\preceq\bar{\mv{G}}_k\preceq \mv{I}}- 2\log \det \left(\bar{\mv{G}}_k\right) -  2\log \det \left(\mv{I}-\bar{\mv{G}}_k\right) \nonumber \\& \sum_{i=2,i\in \mathcal{N}_a^L}^{n} \log\left(-f_{k,i}\mathtt{tr}\left(\bar{\mv{G}}_k \left(\mv{S}_i^{\rm{L}}-\mv{S}^{\rm{L}}_{i-1}\right)\right)\right).\label{eqn:13:multiK:modi}
\end{align}

\begin{figure}
\centering
 \epsfxsize=1\linewidth
    \includegraphics[width=8.8cm]{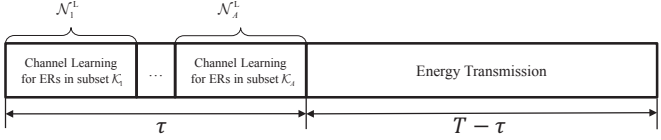}\vspace{0.3em}
\caption{The transmission protocol with feedback interval partition for the ACCPM based multiuser channel learning.} \label{fig:protocol:multiuser}
\end{figure}



Next, we design $\mv{S}_n^{\rm{L}}$'s such that at any feedback interval $n\in \mathcal{N}^{\rm L}_a$ (except $n=1$), the equations in (\ref{eqn:neutral:K}) hold for the subset of ERs in $\mathcal{K}_a$, $a \in\{1,\ldots, A\}$. We set $\mv{S}^{\rm{L}}_1=\frac{P}{M_T}\mv{I}$ for interval $n=1$, and
\begin{align}\label{eqn:A:multi:K}
\mv{S}^{\rm{L}}_n=\mv{S}^{\rm{L}}_{n-1}+\bar{\mv{B}}_n
\end{align}
for the remaining intervals  $n=2,\ldots,N_L$, where $\bar{\mv{B}}_n\in\mathbb{C}^{M_T\times M_T}$ denotes the probing matrix for the multiuser case (as opposed to ${\mv{B}}_n$ in (\ref{eqn:A}) for the single-user case) to be designed such that $\mathtt{tr}(\tilde{\mv{G}}_k^{(n-1)}\bar{\mv{B}}_n)=0, \forall k\in \mathcal{K}_a$, with $n\in \mathcal{N}^{\rm L}_a$. By denoting $\bar{\mv{b}}_n=\mathrm{cvec}\left(\bar{\mv{B}}_n\right)$ and $\tilde{\mv{g}}_k^{(n-1)} = \mathrm{cvec}\left(\tilde{\mv{G}}_k^{(n-1)}\right)$, then finding such a $\bar{\mv{B}}_n$ is equivalent to finding a vector $\bar{\mv{b}}_n$ that is orthogonal to {\it all the $|\mathcal{K}_a|$ vectors},  $\{\tilde{\mv{g}}_k^{(n-1)}\}_{k\in\mathcal{K}_a}$, i.e., $\tilde{\mv{g}}_k^{(n-1)T}\bar{\mv{b}}_n = 0, \forall k\in \mathcal{K}_a$. To do so, we define a $|\mathcal{K}_a| \times M_T^2$ real matrix denoted by $\mv \Phi_n$ with columns composed by the $|\mathcal{K}_a|$ normalized vectors $\left\{\frac{\tilde{\mv{g}}_k^{(n-1)}}{\|\tilde{\mv{g}}_k^{(n-1)}\|}\right\}$ with $k\in\mathcal{K}_a$, based on which we obtain a projection matrix $\bar{\mv{F}}_n = \mv{I} - {\mv \Phi_n{\mv \Phi_n^T}}$. Let  $\bar{\mv{F}}_n = \bar{\mv{V}}_n\bar{\mv{V}}^T_n$, where $\bar{\mv{V}}_n \in \mathbb{R}^{M_T^2\times (M_T^2-|\mathcal{K}_a|+1)}$ satisfies $\bar{\mv{V}}^T_n\mv \Phi_n= \mv{0}$ and $\bar{\mv{V}}^T_n\bar{\mv{V}}_n = \mv{I}$. Then we can find $\bar{\mv{b}}_n$ by setting
\begin{align}\label{learning:3:multi:K}
\bar{\mv{b}}_n = \bar{\mv{V}}_n\bar{\mv{p}},
\end{align}
where $\bar{\mv{p}}\in\mathbb{R}^{(M_T^2-|\mathcal{K}_a|+1)\times 1}$ is a randomly generated vector. With the obtained $\bar{\mv{b}}_n$, we have the probing matrix $\bar{\mv{B}}_n = \mathrm{cmat}(\bar{\mv{b}}_n)$, and accordingly obtain $\mv{S}_n^{\rm{L}}$.{\footnote{Similar to the single-user case, we choose $\|\bar{\mv{p}}\|=P/10$ in order to obtain $\mv{S}_n^{\rm L}$ that satisfies both $\mathtt{tr}(\mv{S}_n^{\rm L})\leq P$ and $\mv{S}_n^{\rm L}\succeq \mv{0}$. }}

\begin{table}[!t]\scriptsize
\caption{ACCPM Based Channel Learning Algorithm for Multiuser Case}
\label{table3} \centering
\begin{tabular}{|p{8.5cm}|}
\hline
\textbf{Algorithm 2}\\
\hline\vspace{0.01cm}
  1) {\bf Initialization:} Set $n=0$, $Q_0^{\rm{L}}=0$, and $\mv{S}_1^{\rm{L}}=\frac{P}{M_T}\mv{I}$; divide the $K$ ERs and $N_L$ feedback intervals into $A = \lceil\frac{K}{M_T^2-1}\rceil$ subsets. \\
2) {\bf Repeat:}
  \begin{itemize} \setlength{\itemsep}{0pt}
    \item[a)] $n \gets n+1$;
    \item[b)] The ET transmits with $\mv{S}_n^{\rm{L}}$;
    \item[c)] Find the user subset index $a$ such that $n+1\in \mathcal{N}^{\rm L}_a$;
    \item[d)] Each ER $k\in\mathcal{K}_a$ feeds back $f_{k,n}=-1$ (or $1$) if $Q_{k,n}^{\rm{L}} \ge Q_{k,n-1}^{\rm{L}}$ (or otherwise);
    \item[e)] The ET computes the query points for all ERs in subset $a$, i.e., $\tilde{\mv{G}}_k^{(n)}$'s, $\forall k\in\mathcal{K}_a$, given in (\ref{eqn:13:multiK:modi});
    \item[f)] The ET computes $\bar{\mv{b}}_{n+1}$ from (\ref{learning:3:multi:K}) based on $\{\tilde{\mv{G}}_k^{(n)}\}_{ k\in\mathcal{K}_a}$, obtains $\bar{\mv{B}}_{n+1}=\mathrm{cmat}(\bar{\mv{b}}_{n+1})$, and updates $\mv{S}^{\rm{L}}_{n+1}=\mv{S}_{n}^{\rm{L}}+\bar{\mv{B}}_{n+1}$.
  \end{itemize}
  3) {\bf Until}  $n\ge N_L$.\\
  4) The ET computes $\tilde{\mv{G}}_k^{(N_L)}$ from (\ref{eqn:13:multiK:modi}) and estimates $\tilde{\mv{G}}_k={\tilde{\mv{G}}^{(N_L)}_k}\big/{\left\|\tilde{\mv{G}}^{(N_L)}_k\right\|_{\rm F}},\forall k\in\mathcal{K}$.\\
 \hline
\end{tabular}\vspace{-1em}
\end{table}


To summarize, we present the ACCPM based channel learning algorithm with one-bit feedback for the multiuser case in Table II as Algorithm 2. 


\subsection{Convergence Analysis: Multiuser Case}\label{channelLearning:multiK}

We provide the convergence analysis for the  ACCPM based multiuser channel learning algorithm in the following proposition.

\begin{proposition}\label{proposition:2:multiK}
The  ACCPM based multiuser channel learning algorithm obtains $K$ MIMO channel estimates for all ERs, i.e., $\{\tilde{\mv{G}}_k\}$, with $\|\tilde{\mv{G}}_k-{\mv{G}}_k\|_{\rm F}\le \varepsilon, \forall k\in\mathcal{K},$ in at most $\mathcal O\left(\left\lceil\frac{K}{M_T^2-1}\right\rceil\frac{M_T^3}{\varepsilon^2} \right)$ number of feedback intervals.
\end{proposition}
\begin{proof}
Given $A=\left\lceil\frac{K}{M_T^2-1}\right\rceil$ partitioned user subsets, $\mathcal{K}_1, \ldots, \mathcal K_A$, each with no more than $M_T^2-1$ ERs, we consider any subset $a\in\{1,\ldots,A\}$. According to (\ref{eqn:neutral:K}) and (\ref{learning:3:multi:K}), at each feedback interval we can simultaneously find $|\mathcal{K}_a| \le M_T^2-1$ neutral cutting planes each for one ER in $\mathcal{K}_a$. As a result, after $\mathcal O\left(\frac{M_T^3}{\varepsilon^2} \right)$ number of feedback intervals, we will have $\mathcal O\left(\frac{M_T^3}{\varepsilon^2} \right)$ neutral cutting planes for each ER in $\mathcal{K}_a$. Based on this argument together with Proposition \ref{proposition:2}, it follows that in at most $\mathcal O\left(\frac{M_T^3}{\varepsilon^2} \right)$ intervals we can have $\|\tilde{\mv{G}}_k-{\mv{G}}_k\|_{\rm F}\le \varepsilon, \forall k\in\mathcal{K}_a$. Given this result and considering that there are in total $A=\left\lceil\frac{K}{M_T^2-1}\right\rceil$ subsets of MIMO channels to be estimated, Proposition \ref{proposition:2:multiK} thus follows.
\end{proof}

Proposition \ref{proposition:2:multiK} provides the worst-case convergence performance for arbitrary values of $M_T$, $M_R$ and $K$. Note that if $K\le M_T^2-1$, it immediately follows from Proposition \ref{proposition:2:multiK} that our proposed algorithm is able to learn $K>1$ ERs' MIMO channels simultaneously without reducing the analytic convergence speed.

\section{Numerical Results}\label{sec:numerical}


\begin{figure}
\centering
 \epsfxsize=1\linewidth
    \includegraphics[width=7cm]{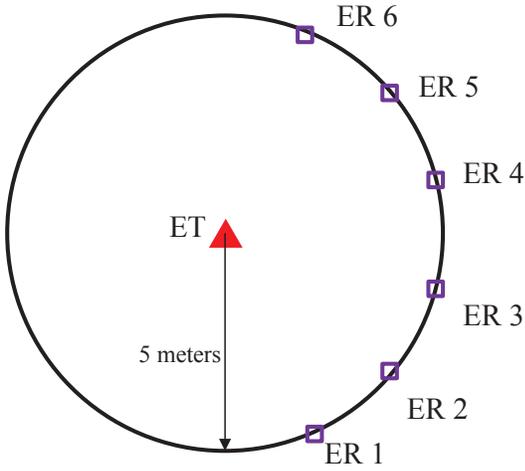}\vspace{0.3em}
\caption{System setup for simulation.} \label{fig:simu_setup}
\end{figure}

In this section, we provide extensive simulation results to evaluate the performance of our proposed ACCPM based channel learning algorithm with one-bit feedback. We consider a multiuser broadcast system for WET as  shown in Fig. \ref{fig:simu_setup}, where 6  ERs are located at an equal distance of 5 meters from the ET, but with different directions. Accordingly, it is assumed that the average path loss from the ET to all ERs is identically 40 dB.  For the considered short transmission distance, the line-of-sight (LOS) signal is dominant, and thus the Rician fading is used to model the channel from the ET to each ER. Specifically, we have
\begin{align}
\mv{H}_k' = \sqrt{\frac{K_R}{1+K_R}}\mv{H}_k^{\rm{LOS}}+\sqrt{\frac{1}{1+K_R}}\mv{H}_k^{\rm{NLOS}},  k\in\mathcal{K},
\end{align}
where $\mv{H}^{\rm{LOS}}_k \in\mathbb{C}^{M_R\times M_T}$ is the LOS deterministic component, $\mv{H}^{\rm{NLOS}}_k \in\mathbb{C}^{M_R\times M_T}$ denotes the non-LOS Rayleigh fading component with each element being an independent circularly symmetric complex Gaussian (CSCG) random variable with zero mean and covariance of $10^{-4}$ (to be consistent with the assumed average power attenuation of $-40$ dB), and $K_R$ is the Rician factor set to be $5$ dB. For the LOS component, we use the far-field uniform linear antenna array model with each row of $\mv{H}^{\rm{LOS}}_k$ expressed as $10^{-2}\left[1~ e^{j\theta_k}~\cdots~e^{j(M_T-1)\theta_k}\right]$ with $\theta_k = -\frac{2\pi \kappa\sin(\phi_k)}{\lambda}$, where  $\kappa$ is the spacing between two successive antenna elements at the ET, $\lambda$ is the carrier wavelength, and $\phi_k$ is the direction of the ER $k$ from the ET. We set $\kappa=\frac{\lambda}{2}$ and $\phi_k = -75^\circ + 30^\circ(k-1), \forall k\in\mathcal{K}$ (see Fig. \ref{fig:simu_setup}). Furthermore, we set the transmit power at the ET and the energy harvesting efficiency at each ER as $P=30$ dBm ($1$ W) and $\varsigma = 50\%$, respectively. In the following, we present our simulation results for the cases of single-user and multiuser WET systems, respectively.

\subsection{Single-User Setup}

\begin{figure}
\centering
 \epsfxsize=1\linewidth
    \includegraphics[width=8.8cm]{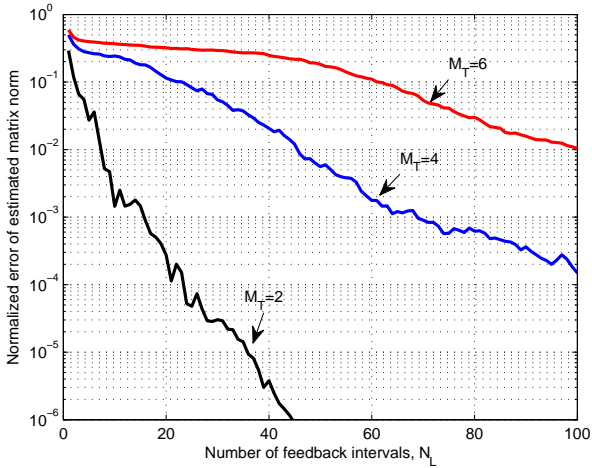}
\caption{Normalized error of estimated matrix norm versus $N_L$ with different number of transmit antennas, $M_T$, where $M_R=2$ and $K=1$.} \label{fig:1}
\end{figure}

\begin{figure}
\centering
 \epsfxsize=1\linewidth
    \includegraphics[width=8.8cm]{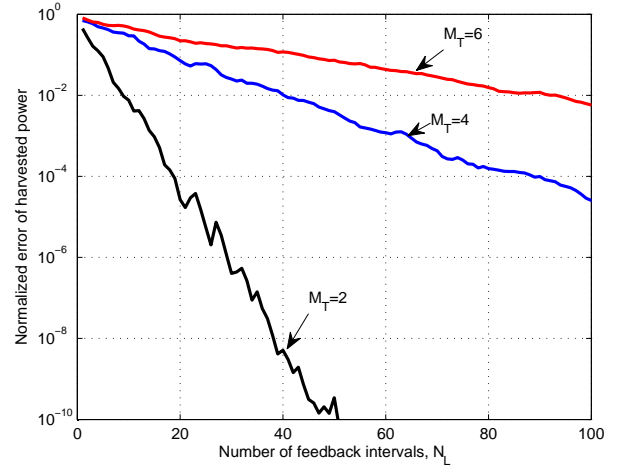}
\caption{Normalized error of harvested power versus $N_L$ with different number of transmit antennas, $M_T$, where $M_R=2$ and $K=1$.} \label{fig:2}
\end{figure}

\begin{figure}
\centering
 \epsfxsize=1\linewidth
    \includegraphics[width=8.8cm]{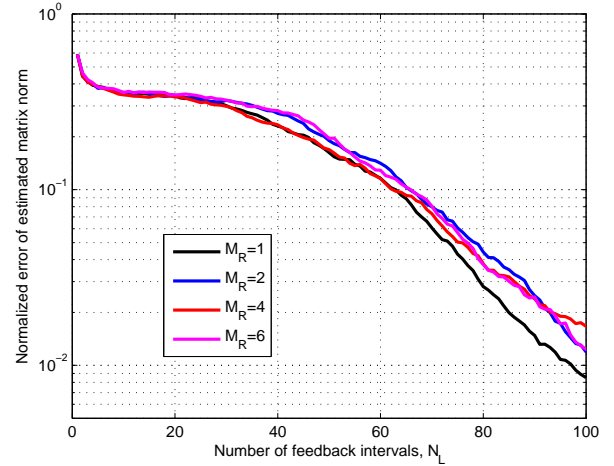}
\caption{Normalized error of estimated matrix norm versus $N_L$ with different number of receive antennas, $M_R$, with fixed $M_T=6$ and $K=1$.} \label{fig:3}
\end{figure}

First, we consider the point-to-point or single-user MIMO WET system with ER  1 only shown in Fig. \ref{fig:simu_setup}. We show the convergence performance of the proposed ACCPM based channel learning algorithm in Figs. \ref{fig:1}-\ref{fig:3}. In Figs. \ref{fig:1} and \ref{fig:2}, we plot the normalized error of estimated matrix norm, i.e., $\|\tilde{\mv{G}} - \mv{G}\|_{\rm F}$, and the normalized error of harvested power, i.e., $\frac{\lambda_E - \tilde{\mv{v}}_E^H\mv{G}\tilde{\mv{v}}_E}{\lambda_E}$, versus the number of feedback intervals in the channel learning phase, $N_L$, with different number of transmit antennas,  $M_T$, and fixed number of receive antennas,  $M_R=2$. From both figures, it is observed that the proposed algorithm achieves an exponentially (or linearly in the log-scale shown in the figures) decreasing error over the number of feedback intervals, which shows its fast convergence in practical implementation. It is also observed that as the number of transmit antennas, $M_T$, increases, the algorithm convergence speed becomes slower. This is due to the fact that for ACCPM, there are in total $M_T^2$ independent real variables in the matrix $\mv{G}$ to be estimated, whose number increases quadratically with $M_T$. Moreover, Fig. \ref{fig:3} shows the normalized error of estimated matrix norm versus $N_L$ with different number of receive antennas, $M_R$, where the number of transmit antennas is fixed as $M_T = 6$. It is observed that the algorithm converges almost at the same speed for different values of $M_R$. This result is consistent with Proposition \ref{proposition:2}, which shows that the analytic convergence speed of the ACCPM based channel learning algorithm is irrelevant to the number of receive antennas, $M_R$.

Fig. \ref{fig:4} compares the channel learning performance of the proposed ACCPM based algorithm against three benchmark algorithms (namely, CJT \cite{Noam2013}, gradient sign \cite{BanisterZeidler2003} and distributed beamforming \cite{Mudumbai2010}; for the details of implementing these algorithms, please refer to Appendix \ref{appendixC}) in terms of the normalized error of harvested power versus $N_L$. It is observed that the CJT algorithm results in discrete error points corresponding to different values of the line search accuracy parameter, $\eta$, and different numbers of sweeps implemented (see Appendix C-1). This is due to the fact that this method can obtain an updated channel estimate only after each complete sweep over a certain number of feedback intervals. For the other two algorithms of gradient sign and distributed beamforming, it is observed that larger step sizes (i.e., $\xi = 0.05$ for gradient sign and $\chi = 0.1\pi$ for distributed beamforming, as defined in Appendices C-2 and C-3, respectively) yield faster convergence speed but also more notable fluctuations as compared to the case of smaller step sizes (i.e., $\xi = 0.01$ for gradient sign and $\chi = 0.04\pi$ for distributed beamforming). In terms of convergence speed, ACCPM is observed to significantly outperform the other three algorithms.

\begin{figure}
\centering
 \epsfxsize=1\linewidth
    \includegraphics[width=8.8cm]{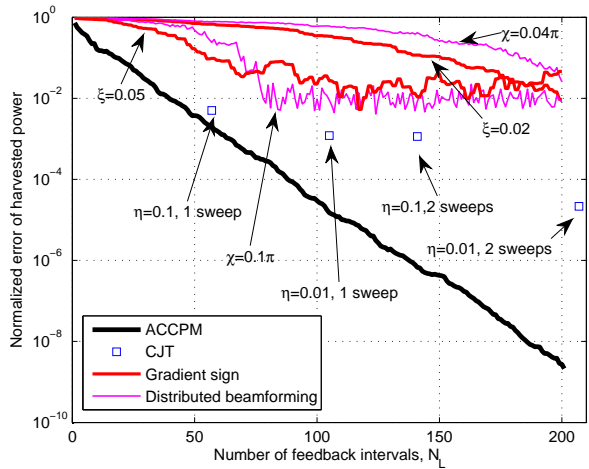}
\caption{Convergence performance comparison for different channel learning algorithms with $M_T=4$, $M_R=2$ and $K=1$.} \label{fig:4}
\end{figure}

\begin{figure}
\centering
 \epsfxsize=1\linewidth
    \includegraphics[width=8.8cm]{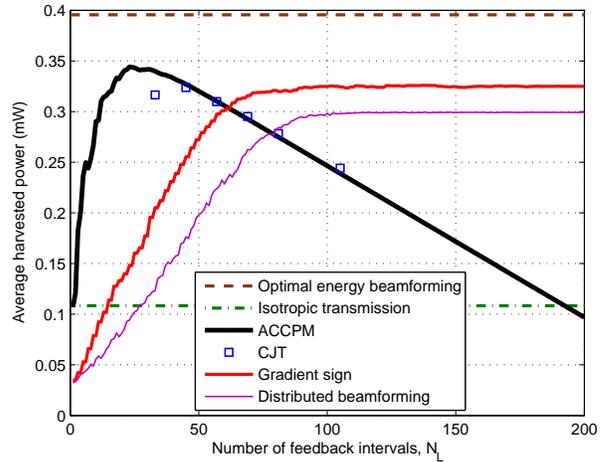}
\caption{Average harvested power comparison for different algorithms with $N=200$, $M_T=4$, $M_R=2$ and $K=1$.} \label{fig:5}
\end{figure}

Fig. \ref{fig:5} shows the average harvested power per block, i.e., $Q_{\rm{total}}/T$ with $Q_{\rm{total}}$ given in (\ref{eqn:Qpro}), versus $N_L$ with fixed block-length, $N=200$, for different algorithms. We set $M_T=4$ and $M_R=2$. For comparison, besides the three benchmark algorithms previously introduced, we also plot the maximum harvested power, $Q_{\max}/T$, by the OEB assuming perfect CSI at the ET as a performance upper bound, as well as the harvested power in the case without CSI at the ET by an isotropic transmission with $\mv{S}=\frac{P}{M_T}\mv{I}$ as a performance lower bound. It is observed that for the ACCPM and CJT based channel learning, the average harvested power first increases and then decreases as $N_L$ increases, and the maximum power value is achieved when $N_L=22$ and $N_L=42$, respectively; while for the gradient sign and distributed beamforming based algorithms, the average harvested power increases consistently with $N_L$. The explanation is as follows. For ACCPM and CJT, the transmit covariance matrices at the ET are in general of full-rank during the channel learning phase, which are designed for channel learning only and thus may not be optimal for energy transmission; therefore, there is in general a trade-off in the time allocations between the channel learning phase versus the energy transmission phase given a fixed finite block-length $N$, to achieve the maximum average harvested power. For gradient sign and distributed beamforming, the transmitted energy beam during the channel learning phase is continuously improved toward the OEB, and thus no energy transmission phase is needed; as a result, the average harvested power increases consistently with increasing $N_L$ until it becomes equal to $N$.

\begin{figure}
\centering
 \epsfxsize=1\linewidth
    \includegraphics[width=8.8cm]{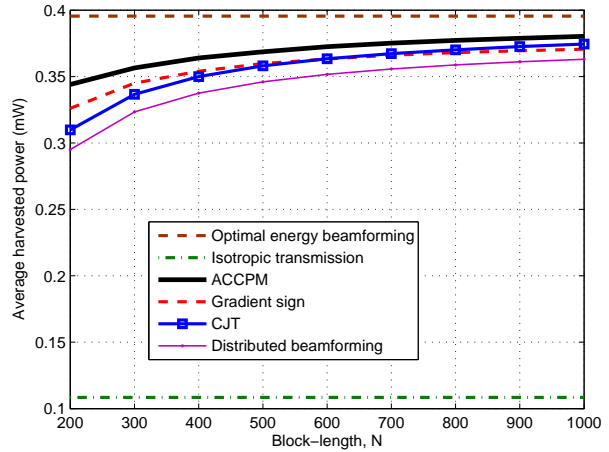}
\caption{Average harvested power  comparison for different algorithms with $M_T=4$ and $M_R=2$.} \label{fig:6}
\end{figure}

Fig. \ref{fig:6} shows the average harvested power per block versus the block-length $N$, where $N_L$ is chosen for each given $N$ and each channel learning algorithm to maximize the corresponding average harvested power. For all the proposed and three benchmark channel learning algorithms, it is observed that as $N$ increases, the average harvested power increases to more closely approach the performance upper bound by the OEB with perfect CSI. This is due to the fact that with larger block-length, the MIMO channel can be estimated more accurately but with smaller percentage of time in each block. The proposed ACCPM based algorithm is observed to achieve higher average harvested power than the other three schemes of CJT, gradient sign and distributed beamforming. This is consistent with its best channel learning performance previously shown in Fig.  \ref{fig:4}.

\subsection{Multiuser Setup}

\begin{figure}
\centering
 \epsfxsize=1\linewidth
    \includegraphics[width=8.8cm]{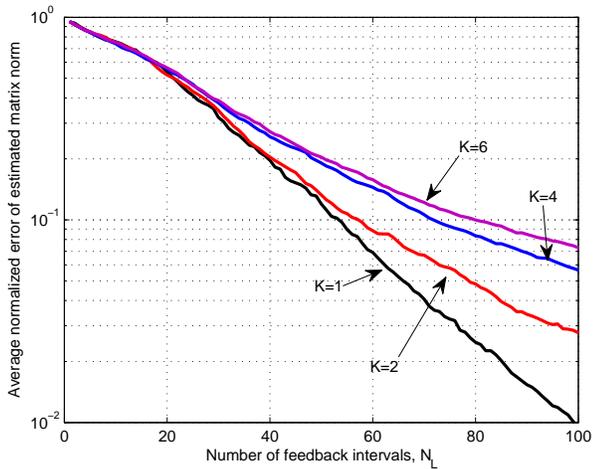}
\caption{Average normalized error of estimated matrix norm versus $N_L$ with different number of ERs,   $K$, where $M_T=4$ and $M_R=2$.} \label{fig:7}
\end{figure}

\begin{figure}
\centering
 \epsfxsize=1\linewidth
    \includegraphics[width=8.8cm]{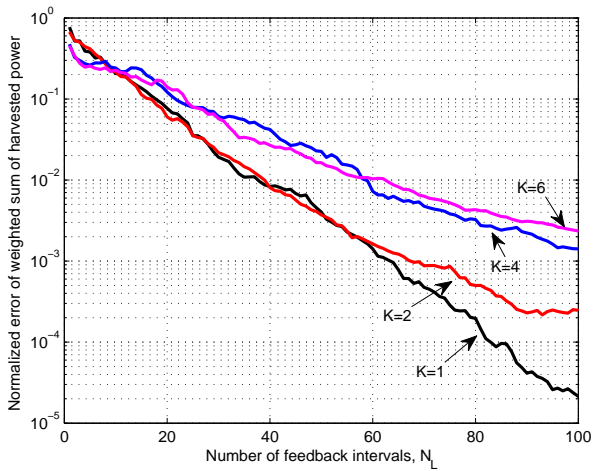}
\caption{Normalized error of weighted sum of harvested power versus $N_L$ with different number of ERs, $K$, where $M_T=4$ and $M_R=2$.} \label{fig:8}
\end{figure}

Next, we consider the multiuser MIMO WET system with $K> 1$ ERs shown in Fig. \ref{fig:simu_setup}. In order to avoid the performance variations due to different user groupings for the proposed multiuser channel learning algorithm (i.e., Algorithm 2), here we focus on the case of $K \le  M_T^2-1$ without the need of user grouping.

Figs. \ref{fig:7} and \ref{fig:8} show the average normalized error of estimated matrix norm, i.e., $\sum_{k\in\mathcal{K}}\|\tilde{\mv{G}}_k - \mv{G}_k\|_{\rm F}/K$, and the normalized error of weighted sum of average harvested power, i.e., $\frac{\lambda_E - \tilde{\mv{v}}_E^H\mv{G}\tilde{\mv{v}}_E}{\lambda_E}$, with different number of ERs, $K$, by fixing $M_T=4$ and $M_R=2$. From both figures, it is observed that as $K$ increases, the algorithm converges more slowly. However, it should be pointed out that this observation does not contradict our result in Section \ref{sec:one-bit:multi} that the  ACCPM based multiuser channel learning algorithm can efficiently estimate multiple ERs' channels at the same time without reducing the analytic convergence speed, provided that $K\le M_T^2-1$, since it is only a worst-case analysis. Moreover, it is observed  from Fig. \ref{fig:8} that after 60 feedback intervals, at least 99\% of the maximum weighted-sum power is achieved for all values of $K$.

\begin{figure}
\centering
 \epsfxsize=1\linewidth
    \includegraphics[width=8.8cm]{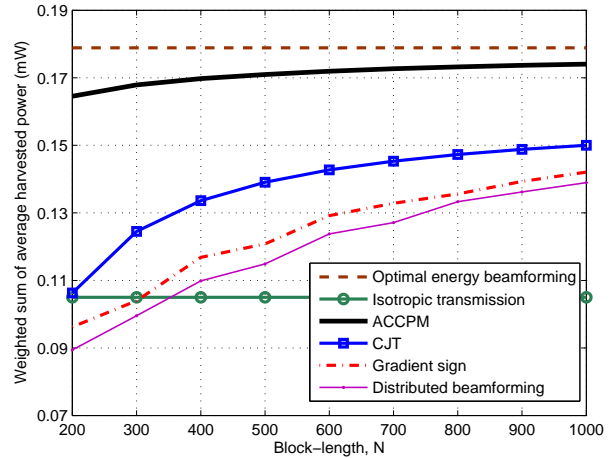}
\caption{Weighted sum of average harvested power comparison for different algorithms, where $M_T=4$, $M_R=2$ and $K=6$.} \label{fig:9}
\end{figure}

Fig. \ref{fig:9} shows the weighted sum of average harvested power per block, i.e., $Q_{\rm{total}}/T$, versus $N$, where $N_L$ is chosen to be the optimal value for each channel learning algorithm and each given $N$. We set $M_T=4$,  $M_R=2$ and $K=6$. Similar to Fig. \ref{fig:6} for the single-user case, it is observed that as $N$ increases, the weighted sum of average harvested power increases for all the proposed and three benchmark algorithms. However, the performance gains of the proposed ACCPM over the three benchmark algorithms become more substantial as compared to that in the single-user case shown in Fig. \ref{fig:6}. This result demonstrates the benefit of ACCPM based channel learning due to simultaneously estimating more than one ERs' MIMO channels, as compared to the three benchmark algorithms that can only estimate the eigenvectors or the dominant eigenmode of one ER's MIMO channel at one time (see Appendix \ref{appendixC} for more details). As a result, it is concluded that our proposed ACCPM based channel learning algorithm is more appealing for the MIMO WET system with multiple ERs.

\section{Concluding Remarks}\label{sec:conclusion}
This paper proposed a new channel learning design approach for multiuser MIMO WET systems. By requiring each ER to send back to the ET only one bit per feedback interval to indicate the increase or decrease of its harvested energy, we show that the ET is able to adjust the training energy beams over different intervals to estimate multiuser MIMO channels simultaneously based on the principle of ACCPM. Through both analysis and simulation, it is shown that our proposed ACCPM based channel learning is more appealing as compared to existing methods with one-bit feedback, especially when the number of ERs is large. It is our hope that this paper will open up an avenue for future investigation of new channel training/feedback techniques for MIMO WET systems.

It is worth noting that in this paper, we considered the one-bit feedback at each ER to simplify the feedback design and receiver complexity and also minimize the energy used for feedback communication. Nevertheless, there may exist alternative feedback designs to further improve the channel learning performance. For example, if each ER can send more than one feedback bits per interval, the channel learning performance should be further improved. How to extend the ACCPM based channel learning with more than one feedback bits per interval is an open problem, which is worth of further investigation.

It is also worth pointing out that although we considered the channel learning for MIMO WET systems in this paper, our proposed ACCPM based algorithm can also be extended to other multiuser MIMO systems in wireless communication for channel learning with low-complexity feedback, e.g., the cognitive radio system considered in \cite{Noam2013} and the precoding design problem studied in \cite{BanisterZeidler2003}.

While preparing this manuscript, the authors become aware of one parallel work \cite{Gopalakrishnan} that was submitted to the same conference as the conference version of this paper \cite{XuJie2014}. \cite{Gopalakrishnan} considered a point-to-point multiple-input single-output (MISO) transmit beamforming system in wireless communication (instead of WET), where ACCPM is applied to estimate the single-user MISO channel at the transmitter based on the one-bit feedback from the receiver indicating whether the received SNR is larger or smaller than a given threshold. Although the core idea of using ACCPM for designing one-bit feedback based channel learning is essentially the same for the two papers, there are still noticeable differences due to independent investigations, which are briefly highlighted as follows. In \cite{Gopalakrishnan}, the transmitter needs to send an additional threshold to the receiver at each training interval to facilitate the one-bit feedback from the receiver; furthermore, the transmitted signal covariance in \cite{Gopalakrishnan}  is of rank-one or corresponds to transmit beamforming for both training and data transmission, while in this paper, we propose to use multi-beam training signals to achieve the faster convergence of channel estimation.

\appendix

\subsection{Proof of Proposition \ref{proposition:1}}\label{appendix:1}

In order to facilitate the proof, for any complex matrix $\mv{A} \in\mathbb{C}^{x \times y}$, we define the following operation that maps $\mv{A}$ into a real matrix as
\begin{align}\label{mapping}
\widehat{\mv{A}} \triangleq \left[\begin{array}{cc}
\mathtt{Re}(\mv{A}) & -\mathtt{Im}(\mv{A})   \\
\mathtt{Im}(\mv{A})  &\mathtt{Re}(\mv{A})
\end{array}\right] \in\mathbb{R}^{2x \times 2y},
\end{align}
where $\mathtt{Re}(\mv{A}) \in\mathbb{R}^{x \times y}$ and $\mathtt{Im}(\mv{A}) \in\mathbb{R}^{x \times y}$ denote the real and imaginary part of $\mv{A}$, respectively. 
\begin{lemma}\label{lemma:mapping}
For any two complex matrices $\mv{A}$ and $\mv{B}$ with appropriate dimensions, the mappings $\mv{A} \to \widehat{\mv{A}}$ and $\mv{B} \to \widehat{\mv{B}}$ have the following properties:
\begin{align}
\det\left(\widehat{\mv{A}}\right) &~~=~ |\det(\mv{A})|^2 \label{eqn:property:1}\\
\mathtt{tr}\left(\widehat{\mv{A}}\widehat{\mv{B}}\right) &~~=~ 2 \mathtt{tr}\left({\mv{A}}{\mv{B}}\right) \label{eqn:property:2}\\
\|\widehat{\mv{A}}\|_{\rm F} &~~=~ \sqrt{2}\|{\mv{A}}\|_{\rm F} \label{eqn:property:3}\\
\widehat{\mv{A}} \succeq \widehat{\mv{B}} & \iff {\mv{A}} \succeq {\mv{B}}. \label{eqn:property:4}
\end{align}
\end{lemma}

\begin{proof}
The properties (\ref{eqn:property:1}) and (\ref{eqn:property:4}) follow from \cite[Lemma 1]{Telatar1999} and \cite[Corollary 2]{Telatar1999}, respectively. The properties (\ref{eqn:property:2}) and (\ref{eqn:property:3}) can be verified by simple matrix manipulations.
\end{proof}

Now, with the defined mapping operation, we are ready to prove Proposition \ref{proposition:1} by showing the convergence performance for a real counterpart of the ACCPM based channel learning algorithm (i.e., Algorithm 1) based on real (versus  complex) matrices. First, we define the real counterpart of Algorithm 1 as an ACCPM based algorithm that aims to find one feasible point in a target set given as follows to estimate the real matrix $\widehat{\mv G} \in \mathbb{R}^{2M_T \times 2M_T}$ mapped from the complex MIMO channel matrix ${\mv G}$.
\begin{align}\label{learning:targetset_real}
\widehat{\mathcal{X}}&=\{\widehat{\bar{\mv{G}}}^\varepsilon |\mv{0}\preceq \widehat{\bar{\mv{G}}}^\varepsilon \preceq \mv{I},~\|\widehat{\bar{\mv{G}}}^\varepsilon-\beta\widehat{\mv{G}}\|_{\mathrm{F}}\le 2\varepsilon,\forall \beta>0\}.
\end{align}
To find a point in the target set $\widehat{\mathcal{X}}$, at each iteration $n \ge 2$, the real counterpart algorithm obtains the following inequality
\begin{align}\label{eqn:cutting:real}
f_n\mathtt{tr}\left(\widehat{{\mv{G}}} \widehat{\mv{B}_{n}}\right)  \le 0,
\end{align}
which corresponds to a cutting plane ensuring that $\widehat{\mv G}$ should lie in the half space of $\widehat{\mathcal{H}}_n=\{\widehat{\bar{\mv{G}}}|f_n\mathtt{tr}\left(\widehat{\bar{\mv{G}}} \widehat{\mv{B}_{n}}\right)  \le 0\}$. Accordingly, the working set $\widehat{\mathcal{P}}_n$ for the real counterpart algorithm can be updated as $\widehat{\mathcal{P}}_n = \widehat{\mathcal{P}}_{n-1}\cap\widehat{\mathcal{H}}_{n}$, with $\widehat{\mathcal{P}}_0=\widehat{\mathcal{P}}_1=\{\widehat{\bar{\mv{G}}} |\mv{0}\preceq \widehat{\bar{\mv{G}}} \preceq \mv{I}\}$. Thus, we have
\begin{align}
\label{learning:workingset_real}
\widehat{\mathcal{P}}_n& =  \left\{\widehat{\bar{\mv{G}}}\big|\mv{0}\preceq\widehat{\bar{\mv{G}}}\preceq \mv{I},~f_i\mathtt{tr}\left(\widehat{\bar{\mv{G}}}\widehat{\mv{B}_i}\right) \le 0, 2\le i\le n\right\}, n\ge 0.
\end{align}
From (\ref{learning:workingset_real}), the analytic center of $\widehat{\mathcal{P}}_{n}$ can be obtained by solving the following problem involving only real matrices \cite{SunTohZhao2002}
\begin{align}
\mathop\mathtt{min}_{\mv{0}\preceq\widehat{\bar{\mv{G}}}\preceq \mv{I}}~&-\log \det \left(\widehat{\bar{\mv{G}}}\right) -  \log \det \left(\mv{I}-\widehat{\bar{\mv{G}}}\right) \nonumber \\ &-\sum_{i=2}^{n} \log\left(-f_i\mathtt{tr}\left(\widehat{\bar{\mv{G}}} \widehat{\mv{B}_i}\right)\right).\label{eqn:13:real}
\end{align}
It can be easily shown by using Lemma \ref{lemma:mapping} that the optimal solution to (\ref{eqn:13:real}) is indeed $\widehat{\tilde{\mv{G}}^{(n)}}$ (recall that $\tilde{\mv{G}}^{(n)}$ is the analytic center of ${\mathcal{P}}_{n}$ for Algorithm 1, defined in (\ref{eqn:13})), which serves as the query point at iteration $n+1$ for the real counterpart algorithm. With such query points, it follows that at the $n$th iteration, $n\ge2$, the cutting plane in (\ref{eqn:cutting:real}) is neutral given $\widehat{\tilde{\mv{G}}^{(n-1)}}$, i.e., $f_n\mathtt{tr}\left(\widehat{\tilde{\mv{G}}^{(n-1)}} \widehat{\mv{B}_{n}}\right)  = 0$, provided that $f_n\mathtt{tr}\left({\tilde{\mv{G}}^{(n-1)}} {\mv{B}_{n}}\right)  = 0$ for Algorithm 1 (cf. (\ref{eqn:neutral}) and (\ref{eqn:A})).

It can be shown that the real counterpart algorithm defined above has the following relationship to Algorithm 1: the analytic centers of $\{{\tilde{\mv{G}}^{(n)}}\}$ converge to a point in ${\mathcal{X}}$ for Algorithm 1 if and only if those of $\{\widehat{\tilde{\mv{G}}^{(n)}}\}$ converge to a point in $\widehat{\mathcal{X}}$ for its real counterpart algorithm. Therefore, it is evident that proving Proposition \ref{proposition:1} is equivalent to showing that in the real counterpart algorithm, $\{\widehat{\tilde{\mv{G}}^{(n)}}\}$ will converge to a point in $\widehat{\mathcal{X}}$ once the iteration index $n$ satisfies the inequality in (\ref{eqn:convergence}).

Next, we prove this argument for the real counterpart algorithm by applying the techniques given in \cite{SunTohZhao2002}. To do so, we define the $n$th potential function as
\begin{align}
\widehat{\varphi}_n\left(\widehat{\bar{\mv{G}}}\right) =& - \log \det \left(\widehat{\bar{\mv{G}}}\right) - \log \det \left(\mv{I}-\widehat{\bar{\mv{G}}}\right) \nonumber \\ &-\sum_{i=2}^{n} \log\left(-f_i\mathtt{tr}\left(\widehat{\bar{\mv{G}}} \widehat{\mv{B}_i}\right)\right), n \ge 0,\label{eqn:potenial_function:real}
\end{align}
and the potential function associated with $\widehat{\mathcal{P}}_n$ as
\begin{align}
\widehat{\varphi}_n^*\left(\widehat{\mathcal{P}}_n\right) & =\mathop\mathtt{min}\limits_{\widehat{\bar{\mv{G}}}\in \widehat{\mathcal{P}}_n}\widehat\varphi_n\left(\widehat{\bar{\mv{G}}}\right).\label{eqn:potential_set:real}
\end{align}
Note that the problem in (\ref{eqn:potential_set:real}) is same as that in (\ref{eqn:13:real}), whose optimal solution is the analytic center $\widehat{\tilde{\mv{G}}^{(n)}}$. Then, based on the potential function, our proof is obtained using the following ideas. First, by assuming $\widehat{\mathcal{P}}_{n} \supset \widehat{\mathcal{X}}$, we establish both the upper and lower bounds of the potential function $\widehat{\varphi}_n^*\left(\widehat{\mathcal{P}}_n\right), \forall n\ge 1$. Then, we show the convergence performance of the real counterpart algorithm by using the fact that $\{\widehat{\tilde{\mv{G}}^{(n)}}\}$ must converge to a point in $\widehat{\mathcal{X}}$ if  the upper and lower bounds contradict each other \cite{SunTohZhao2002}.

First, we establish both the upper and lower bounds of $\widehat{\varphi}_n^*\left(\widehat{\mathcal{P}}_n\right)$ by the following three lemmas, which are derived by using the fact that the cutting plane in (\ref{eqn:cutting:real}) is neutral for each iteration $n\ge 2$. Here, the three lemmas follow directly from \cite{SunTohZhao2002}, and thus we omit their proofs for brevity.
\begin{lemma}\label{lemma:3}
For any working set $\widehat{\mathcal{P}}_{n} \supset \widehat{\mathcal{X}}$, it follows that
\begin{align*}
\widehat\varphi_{n}^*\left(\widehat{\mathcal{P}}_{n}\right) \le -(n+4M_T-1)\log(2\varepsilon), n\ge 1.
\end{align*}
\end{lemma}

\begin{lemma}\label{lemma:1}
$\widehat\varphi^*_n\left(\widehat{\mathcal{P}}_{n}\right)$'s satisfy
\begin{align*}
\widehat\varphi_{n+1}^*\left(\widehat{\mathcal{P}}_{n+1}\right) \ge \widehat\varphi^*_n\left(\widehat{\mathcal{P}}_{n}\right) - \log r_{n+1} + c, \forall n\ge 1,
\end{align*}
where $c > 0$ is a constant, and $r_{n+1}$ is defined as \begin{align}\label{eqn:r_n}
r_{n+1} = \sqrt{({\mathrm{svec}}\widehat{\mv{B}_{n+1}})^T\left(\nabla^2\widehat\varphi_n\left(\widehat{\tilde{\mv{G}}^{(n)}}\right)\right)^{-1}({\mathrm{svec}}\widehat{\mv{B}_{n+1}})}
\end{align}
with $\nabla^2\widehat\varphi_n\left(\widehat{\bar{\mv{G}}}\right)$ denoting the second-order derivative of $\widehat\varphi_n\left(\widehat{\bar{\mv{G}}}\right)$ with respect to $\widehat{\bar{\mv{G}}}$, and ${\mathrm{svec}}(\cdot)$ denoting a linear isometry operation for real symmetric matrices \cite{SunTohZhao2002}.
\end{lemma}

\begin{lemma}\label{lemma:2}
For $r_{i} $ defined in (\ref{eqn:r_n}), it follows that
\begin{align*}
\sum_{i=2}^{n} r_{i}^2 \le 4M_T^2(2M_T+1)\log\left(1+\frac{n-1}{16M_T^2(2M_T+1)}\right).
\end{align*}
\end{lemma}

Next, by combining the above three lemmas, we show that in order for the upper and lower bounds of $\widehat\varphi_n^*\left(\widehat{\bar{\mv{G}}}\right)$ to hold, it follows that
\begin{align}
&\varepsilon^2 \le \nonumber\\
\label{eqn:convergence:opposite} &\frac{M_T+4M_T^2(2M_T+1)\log\left(1+\frac{n-1}{16M_T^2(2M_T+1)}\right)}{4n+16M_T-4} \frac{1}{\exp(\frac{2(n-1)c}{n+4M_T-1})},
\end{align}
which is indeed the reversed inequality of (\ref{eqn:convergence}). From Lemma \ref{lemma:1}, we have
\begin{align*}
\widehat\varphi_{n}^*\left(\widehat{\mathcal{P}}_{n}\right) \ge \widehat\varphi^*_1\left(\widehat{\mathcal{P}}_{1}\right) - \sum_{i=2}^{n}\log r_{i} + (n-1) c,~ n\ge 1.
\end{align*}
Combining this with Lemma \ref{lemma:3}, and using the fact that $\widehat\varphi^*_1\left(\widehat{\mathcal{P}}_{1}\right) = -4M_T\log(\frac{1}{2})$, it thus follows that
\begin{align*}
-(n+4M_T-1)\log(2\varepsilon) \ge -4M_T\log(\frac{1}{2})  - \sum_{i=2}^{n}\log r_{i} + (n -1)c.
\end{align*}
Note that
\begin{align*}
&\frac{4M_T\log(\frac{1}{2}) + \sum_{i=2}^{n}\log r_{i}}{n+4M_T-1} \\ =& \frac{4M_T\log(\frac{1}{4}) + \sum_{i=2}^{n}\log r_{i}^2}{2(n+4M_T-1)} \\ \le& \frac{1}{2}\log\left(\frac{M_T+\sum_{i=2}^{n} r_{i}^2}{n+4M_T-1}\right)\\
 \le &\frac{1}{2}\log\left(\frac{M_T+4M_T^2(2M_T+1)\log\left(1+\frac{n-1}{16M_T^2(2M_T+1)}\right)}{n+4M_T-1}\right),
\end{align*}
where the first inequality follows from the concavity of $\log(\cdot)$ function, and the second inequality is due to Lemma \ref{lemma:2}. As a result, we have
\begin{align*}
&\log(2\varepsilon) \\ \le& \frac{1}{2}\log\left(\frac{M_T+4M_T^2(2M_T+1)\log\left(1+\frac{n-1}{16M_T^2(2M_T+1)}\right)}{n+4M_T-1}\right) \\& -\frac{(n-1)c}{n+4M_T-1}.
\end{align*}
Accordingly, the inequality in (\ref{eqn:convergence:opposite}) follows.



So far, we have proved that in the real counterpart algorithm, given $\widehat{\mathcal{P}}_{n} \supset \widehat{\mathcal{X}}$, in order for the upper and lower bounds of $\widehat{\varphi}_n^*\left(\widehat{\mathcal{P}}_n\right)$ (i.e., Lemmas \ref{lemma:3}-\ref{lemma:2}) to hold, the inequality in (\ref{eqn:convergence:opposite}) must be true. In other words, if the inequality in (\ref{eqn:convergence}) holds or equivalently (\ref{eqn:convergence:opposite}) is violated, then the upper and lower bounds must contradict each other. In this case, it follows from \cite{SunTohZhao2002} that $\{\widehat{\tilde{\mv{G}}^{(n)}}\}$ must converge to a point in $\widehat{\mathcal{X}}$. Based on this result together with the relationship between Algorithm 1 and its real counterpart, it follows that once $n$ satisfies the inequality in (\ref{eqn:convergence}), then $\tilde{\mv{G}}^{(n)}$'s in the ACCPM based single-user channel learning algorithm will converge to a point in the target set $\mathcal{X}$.

Finally, to complete the proof, it remains to show that $\|\tilde{\mv{G}}^{(n)}\|_{\rm F} \ge 1/4$ and the right-hand side in (\ref{eqn:convergence}) is monotonically decreasing with $n \ge 1$. The first argument is true since it follows from (\ref{eqn:13}) that the dominant eigenvalue of $\tilde{\mv{G}}^{(n)}$ is always no smaller than $1/2$, while the second fact can be easily verified by deriving the first-order derivative of the function in the right-hand side in (\ref{eqn:convergence}). As a result, the proof of Proposition \ref{proposition:1} is completed.

\subsection{Proof of Proposition \ref{proposition:2}}\label{appendix:proof2}

First, we show that if ${\tilde{\mv{G}}^{(N_L)}}$ lies in the target set $\mathcal{X}$ with a desired accuracy of $\varepsilon/8$, i.e., there exists a $\beta> 0$ such that $\|{\tilde{\mv{G}}^{(N_L)}}-\beta\mv{G}\|_{\mathrm{F}}\le \varepsilon/8$, then it must hold that $\|\tilde{\mv{G}}-{\mv{G}}\|\le \varepsilon$. Define $\tilde{G}_{\rm F} \triangleq \|\tilde{\mv{G}}^{(N_L)}\|_{\rm F} \ge 1/4$. Then it follows that
\begin{align}
&\left\|{\tilde{\mv{G}}^{(N_L)}}-\beta\mv{G}\right\|_{\mathrm{F}} = { \tilde{G}_{\rm F}}\left\|\frac{\tilde{\mv{G}}^{(N_L)}}{ \tilde{G}_{\rm F}}-\frac{\beta{\mv{G}}}{ \tilde{G}_{\rm F}}\right\|_{\rm F}\nonumber \\  \ge& { \tilde{G}_{\rm F}}\left\|{\mv{G}}-\frac{\beta{\mv{G}}}{ \tilde{G}_{\rm F}}\right\|_{\rm F} = { \tilde{G}_{\rm F}}\left|1- \frac{\beta}{ \tilde{G}_{\rm F}}\right| \ge \frac{1}{4}\left|1- \frac{\beta}{ \tilde{G}_{\rm F}}\right|,\label{eqn:norm:ineqn}
\end{align}
where the first inequality holds due to $\left\|\frac{\tilde{\mv{G}}^{(N_L)}}{ \tilde{G}_{\rm F}}\right\|_{\rm F} = 1$ and $\|{\mv{G}}\|_{\rm F} = 1$. Using (\ref{eqn:norm:ineqn}) together with $\|{\tilde{\mv{G}}^{(N_L)}}-\beta\mv{G}\|_{\mathrm{F}}\le \varepsilon/8$, it follows that $\left|1- \frac{\beta}{ \tilde{G}_{\rm F}}\right| \le \varepsilon/2$. Therefore, it must hold that
\begin{align*}
&\|\tilde{\mv{G}}-{\mv{G}}\|_{\rm F}=\left\|\frac{\tilde{\mv{G}}^{(N_L)}}{ \tilde{G}_{\rm F}}-{\mv{G}}\right\|_{\rm F}\\ \le & \left\|\frac{\tilde{\mv{G}}^{(N_L)}}{ \tilde{G}_{\rm F}} - \frac{\beta{\mv{G}}}{ \tilde{G}_{\rm F}}\right\|_{\rm F} + \left\|\frac{\beta{\mv{G}}}{ \tilde{G}_{\rm F}}-{\mv{G}}\right\|_{\rm F}\le \frac{\varepsilon}{2}+\left|1- \frac{\beta}{ \tilde{G}_{\rm F}}\right|\le  \varepsilon.
\end{align*}
As a result, we have proved that if $\|{\tilde{\mv{G}}^{(N_L)}}-\beta\mv{G}\|_{\mathrm{F}}\le \varepsilon/8$, then $\|\tilde{\mv{G}}-{\mv{G}}\|_{\rm F}\le \varepsilon$.

Next, we show that the proposed algorithm converges to a point ${\tilde{\mv{G}}^{(N_L)}}$ with $\left\|{\tilde{\mv{G}}^{(N_L)}}-\beta{\mv{G}}\right\|_{\rm F}\le \varepsilon/8$ after a number of
$N_L =\mathcal O\left(\frac{M_T^3}{\varepsilon^2} \right)$ feedback intervals. By using Proposition \ref{proposition:1}, we have that in the worst case, the ACCPM based algorithm will converge to a matrix ${\tilde{\mv{G}}^{(N_L)}}$ with $\left\|{\tilde{\mv{G}}^{(N_L)}}-\beta{\mv{G}}\right\|_{\rm F}\le \varepsilon/8$ once
\begin{align*}
\varepsilon^2/64 \ge &\frac{M_T+4M_T^2(2M_T+1)\log\left(1+\frac{N_L-1}{16M_T^2(2M_T+1)}\right)}{(4N_L+16M_T-4)\exp(\frac{2(N_L-1)c}{n+4M_T-1})}.
\end{align*}
By ignoring the lower order terms, we can have $\frac{N_L}{\log\left({N_L}/{M_T^3}\right)} \ge\mathcal O\left(\frac{M_T^3}{\varepsilon^2} \right)$.
Since $\log(N_L)$ is negligible compared to $N_L$, it then follows that ${N_L} \thickapprox\mathcal O\left(\frac{M_T^3}{\varepsilon^2} \right)$.

By combining the above two arguments, Proposition \ref{proposition:2} is thus proved.

\subsection{Other Algorithms for Comparison}\label{appendixC}
For comparison, in this appendix, we introduce three alternative algorithms, namely, CJT \cite{Noam2013}, gradient sign \cite{BanisterZeidler2003}, and distributed beamforming \cite{Mudumbai2010}, for the design of one-bit feedback based channel learning, which were proposed and studied in other application scenarios instead of WET. Note that these algorithms were originally designed to learn one single MIMO/MISO channel only. In order to extend them to be applied for the multiuser MIMO WET system of our interest with $K\ge 1$ ERs, we consider the transmission protocol as shown in Fig. \ref{fig:protocol:existing} for all these three algorithms, in which the channel learning phase is divided into $K$ equal-duration slots. In each slot of length $\tau/K$ (or equivalently $N_L/K$ in number of feedback intervals by assuming $N_L/K$ to be an integer), the ET learns one of the $K$ MIMO channels by applying one from the three channel learning algorithms. In the following, we discuss the implementation details of the three algorithms, respectively. For each algorithm, we first introduce the channel training and estimation method during each slot of the channel learning phase assuming the same (or slightly modified) one-bit feedback scheme at each ER as in our proposed channel learning algorithm, and then explain the energy beamforming design for the energy transmission phase.

\begin{figure}
\centering
 \epsfxsize=1\linewidth
    \includegraphics[width=8.8cm]{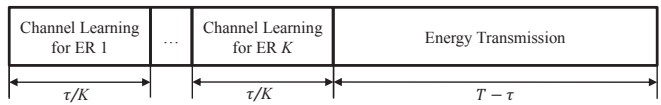}\vspace{0.3em}
\caption{The transmission protocol for the three existing channel learning algorithms.} \label{fig:protocol:existing}
\end{figure}

\subsubsection{CJT \cite{Noam2013}}\label{appendixC1}
First, consider the $k$th channel learning slot for estimating the channel from the ET to the $k$th ER, $k\in\mathcal K$. During this slot, the ET updates $\mv{S}^{\rm L}_n$'s based on the feedback bit $f_{k,n}$'s to implement CJT for blindly obtaining the EVD for $\mv G_k$, i.e., $\mv{G}_k = \mv V_k \mv \Lambda_k \mv{V}_k^H$; accordingly, the ET can obtain an estimate of $\mv{V}_k$ at the end of slot $k$, denoted by $\tilde{\mv V}_k$.

The idea of CJT is based on the fact that the EVD of $\mv G_k$ can be implemented through a series of two-dimensional rotations, where at each rotation one pair of the matrix's two off-diagonal elements are eliminated. Specifically, let $a \ge 1$ denote the rotation index, and define $\mv{A}_0 \triangleq \mv{G}_k$. Then at each rotation $a$, the CJT aims to find an $M_T \times M_T$ unitary rotation matrix $\mv{W}_a$ to construct a new matrix $\mv{A}_{a} = \mv{W}_a\mv{A}_{a-1}\mv{W}_a^H$, such that the $(l_a,m_a)$th and $(m_a,l_a)$th entries of $\mv{A}_{a}$ are both zero, where $1\le l_a\le M_T$ and $l_a < m_a\le M_T$. Note that for the $M_T\times M_T$ Hermitian matrix $\mv{G}_k$, there are $M_T(M_T-1)/2$ pairs of $(l,m)$ satisfying $1\le l\le M_T$ and $l< m\le M_T$; hence, in each set of $M_T(M_T-1)/2$ rotations, the CJT will choose $(l_a,m_a)$ such that each pair of $(l,m)$ is selected once only. We refer to such a set of $M_T(M_T-1)/2$ rotations as one sweep. After each complete sweep, the ET can obtain an improved estimate of $\mv{V}_k$. To successfully implement CJT based on the one-bit feedback, it is key to estimate $\mv{W}_a$ at each rotation $a$. This is done by calculating two real variables via two binary line searches, where each line search is accomplished with several feedback intervals required by utilizing the feedback bit $f_{k,n}$'s (see more details in \cite{Noam2013}).

Next, with the estimated $\tilde{\mv V}_k$'s for all the $K$ ERs after the completion of the channel learning phase, we consider the energy beamforming design for the energy transmission phase. Since the ET can only obtain the estimates of $\{\mv{V}_k\}$ but not those of $\{\mv{\Lambda}_k\}$, it is difficult to reconstruct the estimate of $\mv{G}=\sum_{k\in\mathcal{K}} \mv{G}_k$ for implementing the OEB. To overcome this difficulty, we propose a suboptimal energy beamforming design for WET. Specifically, we define $\bar{\mv G} \triangleq \sum_{k\in\mathcal K} \tilde{\mv{v}}_{k,1}^H \tilde{\mv{v}}_{k,1}$, where $\tilde{\mv{v}}_{k,1}$ denotes the estimate of the dominant eigenvector of $\mv{G}_k$ (i.e., ${\mv{v}}_{k,1}$), which can be obtained based on the estimated $\tilde{\mv{V}}_k, k\in\mathcal K$. Then the ET uses the dominant eigenvector of $\bar {\mv G}$ as the energy beamforming vector $\tilde{\mv v}_E$ during the energy transmission phase, i.e., $\mv{S}^{\rm{E}} = P\tilde{\mv{v}}_E\tilde{\mv{v}}_E^H$. It is worth noting that although this energy beamforming design is suboptimal in general, it is optimal when $K=1$ or $M_R=1$, since in this case the estimated energy beamforming vector will approach to the OEB with perfect CSI at the ET as $N_L\rightarrow \infty$.

It should be pointed out that in the CJT algorithm, the performance of channel learning is controlled by the number of sweeps implemented as well as the desired accuracy of each line search, given by $\eta > 0$. If $\eta$ is set small and the number of sweeps is chosen to be large, the estimation of $\mv{V}_k$ can be made more accurate, but at the cost of more feedback intervals required.

\subsubsection{Gradient Sign  \cite{BanisterZeidler2003}}\label{appendixC2}
First, consider the $k$th slot in the channel learning phase, $k\in\mathcal K$, for which the feedback intervals are indexed by 1 to $N_L/K$ for convenience. During this slot, the ET sends only one energy beam per feedback interval denoted by $\mv{w}_n^{\rm{L}}$ (i.e., $\mv{S}_n^{\rm{L}}  = \mv{w}_n^{\rm{L}}\mv{w}_n^{{\rm L}H}$ is of rank-one), $n=1,\ldots,N_L/K$. By adjusting $\mv{w}_n^{\rm{L}}$'s over different feedback intervals in slot $k$, the ET can obtain an estimate of the dominant eigenvector of $\mv G_k$, i.e., $\mv v_{k,1}$, denoted by $\tilde{\mv v}_{k,1}$, for ER $k$.

The gradient sign algorithm during each slot $k$ is explained as follows by assuming $N_L/K$ to be an even integer for convenience. Define a sequence of reference beamforming vectors for this slot as $\{\tilde{\mv{w}}_{a}\}$ with $0\le a\le \frac{N_L}{2K}-1$, where $\tilde{\mv{w}}_{0} = \sqrt{\frac{P}{M_T}}\mv{1}$ denotes the initial reference beamforming vector. Then, over each odd interval $n=2a+1$, the ET sets the energy beam $\mv{w}_{2a+1}^{\rm{L}}$ by adding a random perturbation vector to $\tilde{\mv{w}}_{a}$, while over the next even interval $n=2a+2$, the ET sets $\mv{w}_{2a+2}^{\rm{L}}$ by subtracting the same random perturbation to $\tilde{\mv{w}}_{a}$. Based on the feedback  $f_{k,2a+2}$ from ER $k$ at the end of interval $n=2a+2$, the ET can know whether $\mv{w}_{2a+1}^{\rm{L}}$ or $\mv{w}_{2a+2}^{\rm{L}}$ achieves a higher transferred energy to ER $k$, and accordingly, it updates the reference beamforming vector $\tilde{\mv{w}}_{a+1}$ as the better one. By performing the above procedure, it is shown in \cite{BanisterZeidler2003} that as $a\to \infty,$ the reference beamforming vector $\{\tilde{\mv{w}}_a\}$ will approach the direction of the dominant eigenvector $\mv v_{k,1}$. Therefore, we can obtain an estimate of $\mv v_{k,1}$ as $\tilde{\mv v}_{k,1}$ over finite number of feedback intervals in slot $k$.

Next, consider the energy transmission phase. Similar to the case of  CJT algorithm, since the ET can only obtain the estimates of ${\mv v}_{k,1}$'s instead of $\mv G_k$'s, the ET applies the dominant eigenvector of $\bar{\mv G} = \sum_{k\in\mathcal K} \tilde{\mv{v}}_{k,1}^H \tilde{\mv{v}}_{k,1}$ as the energy beamforming vector $\tilde{\mv{v}}_E$ for WET.

Note that in the gradient sign algorithm, the performance of channel learning is affected by the norm of the added/subtracted random perturbation vector, which is referred to as the step size $\xi$. In general, larger value of $\xi$ corresponds to faster convergence of the algorithm but also results in more estimation errors.

\subsubsection{Distributed Beamforming \cite{Mudumbai2010}}\label{appendixC3}

Note that in the case of distributed beamforming, the design of one-bit feedback $f_{k,n}$ from ER $k$ needs to be slightly modified as compared to that of the ACCPM, CJT and gradient sign algorithms, i.e., it is designed to indicate whether the harvested energy at ER $k$ over the $n$th interval, $n=1,\ldots,N_L/K$, is larger or smaller than its record of the highest harvested energy so far during the $k$th slot of the channel learning phase. Here, the intervals over slot $k$ are indexed by 1 to $N_L/K$ for convenience, similar to the previous case of gradient sign.

With the modified feedback design, we provide the details of the channel learning for distributed beamforming in each slot $k$. At each interval $n$ of the $k$th slot, $1\le n\le \frac{N_L}{K}$, the ET sends one energy beam denoted by $\mv{w}_n^{\rm{L}} = \sqrt{\frac{P}{M_T}}\left[e^{j\vartheta^{\rm L}_{1,n}} \cdots e^{j\vartheta^{\rm L}_{M_T,n}} \right]^T$, where only the signal phases $\{\vartheta^{\rm L}_{m,n}\}$ are adjustable over different intervals of $n$. We set $\mv{w}_{1} = \sqrt{\frac{P}{M_T}}\mv{1}$ as the initial energy beam. We also define $\tilde{\mv{w}}_n  = \sqrt{\frac{P}{M_T}}\left[e^{j\vartheta_{1,n}} \cdots e^{j\vartheta_{M_T,n}} \right]$ as the best transmit beamforming vector for ER $k$ during the $k$th slot prior to interval $n$, with $\tilde{\mv{w}}_{1} =\mv{w}_{1}^{\rm{L}} = \sqrt{\frac{P}{M_T}}\mv{1}$. At interval $n+1$, $n=1,\ldots,N_L/K-1$, the ET applies a set of random phase perturbations $\delta_{m,n+1}$ to $\vartheta_{m,n}$, $m=1,\ldots,M_T$, in order to probe for a potentially better phase array, i.e., $\vartheta^{\rm L}_{m,n+1} = \delta_{m,n+1} + \vartheta_{m,n}, m=1,\ldots,M_T$. Then, if the feedback bit $f_{k,n+1}$ from the $k$th ER indicates an improvement in the harvested energy, the ET updates $\tilde{\mv{w}}_{n+1}\gets{\mv{w}}_{n+1}^{\rm L}$; otherwise, the ET sets $\tilde{\mv{w}}_{n+1}\gets\tilde{\mv{w}}_n$. By performing the above procedure, the ET can obtain the updated $\tilde{\mv{w}}_{N_L/K}$  at the end of the $k$th channel learning slot, the direction of which is used as an estimate of ${\mv{v}}_{k,1}$, denoted by $\tilde{\mv{v}}_{k,1}$. Note that for the distributed beamforming based algorithm, how to set the phase perturbation $\delta_{m,n}$'s is an important issue. In this paper, we generate $\delta_{m,n}$'s for each interval $n$ randomly based on a uniform distribution over the interval $[-\chi/2,\chi/2]$, where $\chi > 0$ is the step size that controls the algorithm accuracy and speed.

Next, consider the energy transmission phase. Similar to the previous two cases, based on the estimated $\{\tilde{\mv v}_{k,1}\}$ during the channel learning phase, the ET applies the dominant eigenvector of $\bar{\mv G} = \sum_{k\in\mathcal K} \tilde{\mv{v}}_{k,1}^H \tilde{\mv{v}}_{k,1}$ as the energy beamforming vector $\tilde{\mv{v}}_E$ for WET.

\end{document}